\documentclass[12pt]{article}
\usepackage{amsmath}
\usepackage{graphicx}
\usepackage{natbib}
\usepackage{url} 

\usepackage[utf8]{inputenc} 
\usepackage[T1]{fontenc}    
\usepackage{hyperref}       
\usepackage{url}            
\usepackage{booktabs}       
\usepackage{amsfonts}       
\usepackage{nicefrac}       
\usepackage{microtype}      
\usepackage{xcolor}         
\usepackage{setspace}
\usepackage{graphicx}
\usepackage{subfigure}
\usepackage{adjustbox}
\usepackage{paralist}

\usepackage{amsmath}
\usepackage{amssymb}
\usepackage{mathtools}
\usepackage{amsthm}
\usepackage{bbm}

\theoremstyle{plain}
\newtheorem{theorem}{Theorem}[section]
\newtheorem{proposition}{Proposition}
\newtheorem{lemma}{Lemma}
\newtheorem{corollary}{Corollary}
\newtheorem{definition}{Definition}
\newtheorem{assumption}{Assumption}
\newtheorem*{remark}{Remark}

\usepackage[textsize=tiny]{todonotes}
\usepackage{enumitem}

\newcommand{\indep}{\perp \!\!\! \perp}
\newcommand{\E}{\mathbb{E}}

\newcommand{\pr}{P}
\newcommand{\qeps}{q_{\epsilon_w}}
\newcommand{\pbar}{\Bar{p}}
\newcommand{\qbareps}{\Bar{q}_{\epsilon_w}}
\newcommand{\pq}{p\qeps + \pbar\qbareps}
\newcommand{\pqbar}{\pbar\qeps + p\qbareps}

\newcommand{\indsim}{\mathop{\sim}^{ind}}
\newcommand{\iidsim}{\mathop{\sim}^{i.i.d}}

\newcommand{\convp}{\overset{p}{\to}}
\newcommand{\convd}{\overset{D}{\to}}
\newcommand{\Var}{\mathbb{V}{\rm ar}}
\newcommand{\Cov}{\mathbb{C}{\rm ov}}

\DeclareMathOperator*{\argmin}{arg\,min}

\renewenvironment{itemize}[1]{\begin{compactitem}#1}{\end{compactitem}}
\renewenvironment{enumerate}[1]{\begin{compactenum}#1}{\end{compactenum}}

\usepackage{bibspacing}
\newcommand{\blind}{0}

\begin{document}

\def\spacingset#1{\renewcommand{\baselinestretch}%
{#1}\small\normalsize} \spacingset{1}


\if0\blind
{
  \title{\bf Locally Private Causal Inference for Randomized Experiments}
  \author{Yuki Ohnishi\hspace{.2cm}\\
    Department of Statistics, Purdue University\\
    and \\
    Jordan Awan \thanks{
    This work was supported in part by the National Science Foundation (NSF) grants SES
2150615.}\\
    Department of Statistics, Purdue University}
  \maketitle
} \fi

\if1\blind
{
  \bigskip
  \bigskip
  \bigskip
  \begin{center}
    {\LARGE\bf Locally Private Causal Inference for Randomized Experiments}
\end{center}
  \medskip
} \fi

\bigskip
\begin{abstract}
Local differential privacy is a differential privacy paradigm in which individuals first apply a privacy mechanism to their data (often by adding noise) before transmitting the result to a curator. The noise for privacy results in additional bias and variance in their analyses. Thus it is of great importance for analysts to incorporate the privacy noise into valid inference.
In this article, we develop methodologies to infer causal effects from locally privatized data under randomized experiments. First, we present frequentist estimators under various privacy scenarios with their variance estimators and plug-in confidence intervals. We show a na\"ive debiased estimator results in inferior mean-squared error (MSE) compared to minimax lower bounds. In contrast, we show that using a customized privacy mechanism, we can match the lower bound, giving minimax optimal inference. We also develop a Bayesian nonparametric methodology along with a blocked Gibbs sampling algorithm, which can be applied to any of our proposed privacy mechanisms, and which performs especially well in terms of MSE for tight privacy budgets. Finally, we present simulation studies to evaluate the performance of our proposed frequentist and Bayesian methodologies for various privacy budgets, resulting in useful suggestions for performing causal inference for privatized data.
\end{abstract}

\noindent%
\vfill

\spacingset{1.} 


\section{Introduction}
\label{sec:introduction}
Causal inference is a fundamental consideration across a wide range of domains in science, technology, engineering, and medicine. Researchers study experimental or observational data to unveil the causal effects of treatment assignment in an unbiased manner with valid uncertainty quantification. A traditional gold standard for performing causal inference is the classical randomized experiment \citep{imbens_rubin_2015}. In this type of experiment, a great deal of control and precautions can be taken so as to eliminate events that would introduce instabilities and biases in causal inferences. 

On the other hand, differential privacy (DP), introduced by \citet{dwork2006calibrating}, is another growing domain in science and business, as privacy protection has become a core concern for many organizations in the modern data-rich world.  DP is a mathematical framework that provides a probabilistic guarantee that protects private information about individuals when publishing statistics about a dataset. This probabilistic guarantee is often achieved by adding random noise to the data. 
One DP model is the \emph{central} differential privacy model, in which the data curators have access to the sensitive data and apply a DP mechanism to the data to produce the published outputs. A weakness of this model is that users are required to trust the data curators with their sensitive data. Another DP model is \emph{local} differential privacy (LDP). In this model, the users do not directly provide their data to the data curator; instead, users apply the DP mechanism to their data locally before sending it to the curator. LDP is a preferable model if the data curators are not trusted by users. The LDP model has been adopted by various tasks and organizations, e.g., Google \citep{erlingsson2014} and Apple \citep{apple2017}, for more stringent privacy protection.

Drawing causal conclusions from privatized data can be challenging. While the added random noise helps in safeguarding individuals' privacy, it distorts the actual patterns in the data. This distortion can lead to biased conclusions even in randomized experiments. This issue becomes even more pronounced in the LDP method, where each data point is individually altered before it is compiled. Therefore, when trying to understand cause-and-effect relationships using this protected data, researchers must exercise extra caution to ensure their interpretations remain accurate and unbiased.

In this article, we propose statistically valid causal inferential methodologies under three distinct local privacy scenarios. The first scenario, which we refer to as a ``joint scenario'', assumes that all accessible variables are separately privatized. 
In the second and third scenarios, which we term as ``custom scenarios'', we are allowed to select the variables we privatize with known and unknown treatment assignment probabilities. We then offer causal inference methodologies to analyze such privatized data.  Our main contributions are as follows:
\begin{itemize}
    \item We propose a ``na\"ive''  inverse probability weighting (IPW) estimator under the joint scenario. We compute the bias of the IPW estimator and propose a debiasing technique. 
    \item We propose efficient frequentist estimators that achieve the minimax optimal rate under custom scenarios where we are allowed to select the variables we privatize.
    \item We also compute the asymptotic variance and construct asymptotic plugin nominal confidence intervals for all frequentist estimators. We discuss their optimality under each scenario.
    \item We develop a flexible and efficient Bayesian nonparametric methodology, along with a data augmentation Gibbs sampler tailored for locally privatized observations, which can be applied to all scenarios that are considered in the frequentist analyses.
    \item We present simulation studies and empirical data analysis to evaluate the frequentist and Bayesian methodologies at various privacy budgets, resulting in useful suggestions for performing causal inference for privatized data.
    \item We propose a regression adjustment technique under the joint scenario in the Supplementary Materials. We show both theoretically and empirically that it helps improve the accuracy, but the gain is somewhat limited when the privacy budgets are tight.
\end{itemize}
The rest of the paper is organized as follows. 
Section \ref{sec:preliminaries} presents the preliminaries for the Rubin Causal Model and LDP. In Section \ref{sec:frequentist}, we develop frequentist approaches to inferring the causal effects of interest. Section \ref{sec:bayesian} presents a Bayesian methodology for performing valid causal inference with the privatized data. 
Section \ref{sec:simulationstudies} provides simulation studies for validating our methodologies developed in the previous sections, and Section \ref{sec:empirical_analysis} provides an application of our methodologies to real-world data of a cash transfer program conducted in Columbia. Section \ref{sec:conclusion} concludes with some final discussion. The Supplementary Materials contains proofs, technical details, and additional numerical results.

\subsection{Related Work}
\label{sec:relatedwork}
While DP is a rapidly growing field, the literature on causal inference methodologies for differentially privatized data remains sparse. The following work uses LDP for its DP mechanism.
\citet{Agarwal2021} introduced an end-to-end procedure for covariates cleaning, estimation, and inference, offering covariates cleaning-adjusted confidence intervals under the local differential privacy mechanism. 

Some researchers have developed causal inference methodologies under the central DP model.  
\citet{D'Orazio2015} introduced the construction of central differential privacy mechanisms for summary statistics in causal inference. They then presented new algorithms for releasing differentially private estimates of causal effects and the generation of differentially private covariance matrices from which any least squares regression may be estimated. 
\citet{Lee2019} proposed a privacy-preserving inverse propensity score estimator for estimating the average treatment effect (ATE).
\citet{Komarova2020} studied the impact of differential privacy on the identification of statistical models and demonstrated identification of causal parameters failed in regression discontinuity design under the central differential privacy.
\citet{niu2022} introduced a general meta-algorithm for privately estimating conditional average treatment effects. 
\citet{Kusner2016PrivateCI} tackles causal inference using a framework called the additive noise model (ANM), a more restrictive causal model than the Rubin Causal Model.


In non-causal domains, \citet{Evans2022} offered statistically valid linear regression estimates and descriptive statistics for locally private data that can be interpreted as ordinary analyses of non-confidential data but with appropriately larger standard errors. 
\citet{Schein2019} presented an MCMC algorithm that approximates the posterior distribution over the latent variables conditioned on data that has been locally privatized by the geometric mechanism. \citet{Nianqiao_Jordan_2022} proposed a general privacy-aware data augmentation MCMC framework to perform Bayesian inference from privatized data. 


\section{Preliminaries}
\label{sec:preliminaries}
\subsection{Rubin Causal Model}
\label{sec:RCM}
Causal inference is of fundamental importance across many scientific and engineering domains that require informed decision-making based on experiments. Throughout this manuscript, we adopt the Rubin Causal Model (RCM) as our causal paradigm. In the RCM it is critical to first carefully define the Science of a particular problem, i.e., to define the experimental units, covariates, treatments, and potential outcomes \citep{imbens_rubin_2015}. We consider $N$ experimental units, indexed by $i = 1, \ldots , N$, that correspond to physical objects at a particular point in time. Each unit $i$ has an observed outcome $Y_i$ and treatment assignment $W_i$ respectively. We consider a binary treatment $W_i \in \{0,1\}$ with a fixed assignment probability, $p=\pr(W_i=1)$, which is assumed to be known by the experimental design, and let $Y_i(w)$ denote a potential outcome for $w \in \{0,1\}$. 
In this article, we consider the $N$ units as a random sample from a large super-population, and we are interested in inferring the Population Average Treatment Effect (PATE): $\tau = \E[Y_i(1)-Y_i(0)]$.
We invoke the common set of assumptions, which enable us to identify the PATE by the estimators derived in this manuscript \citep{imbens_rubin_2015}.

\begin{assumption}\label{asmp:common}
\begin{enumerate}
\item (Positivity) The probability of treatment assignment given the covariates is bounded away from zero and one: $0<\pr(W_i=1)<1$.
\item (Random Assignment) The potential outcomes are independent of treatment assignment: $\{Y_i(0), Y_i(1)\} \indep W_i$.
\item (Stable Unit Treatment Value Assumption [SUTVA]) There is neither interference nor hidden versions of treatment. The observed outcome is formally expressed as: $Y_i = W_iY_i(1) + (1-W_i)Y_i(0)$.
\end{enumerate}
\end{assumption}

\subsection{Differential Privacy}
\label{sec:DP}
In this article, we use the local differential privacy (LDP) model. 
Let $\mathcal{D}$ be the set of possible contributions from one individual in database $D$. In this paper, we only consider non-interactive local DP mechanisms. LDP is formally defined for any $\mathcal{D}$ as follows.
\begin{definition}[Local Differential Privacy]
    An algorithm $\mathcal{M}$ is said to be $\epsilon$-locally differentially private ($\epsilon$-LDP) if for any two data points $x,x' \in \mathcal{D}$, and any $S \subseteq \mathrm{Range}(\mathcal{M})$,
    \begin{equation*}
        \pr(\mathcal{M}(x) \in S) \leq \exp(\epsilon)\pr(\mathcal{M}(x') \in S).
    \end{equation*}
\end{definition}

Intuitively, if an individual were to change their value from $x$ to $x'$, the output distribution of $M$ would be similar, making it difficult for an adversary to determine whether $x$ or $x'$ was the true value. The value $\epsilon$ is called the \emph{privacy budget} and lower values indicate a stronger privacy guarantee. Two important properties of differential privacy are \emph{composition} and \emph{invariance to post-processing}. Composition allows one to derive the cumulative privacy cost when releasing the results of multiple privacy mechanisms: if $\mathcal M_1$ is $\epsilon_1$-LDP and $\mathcal M_2$ is $\epsilon_2$-DP, then the joint release $(\mathcal M_1(x),\mathcal M_2(x))$ satisfies $(\epsilon_1+\epsilon_2)$-LDP. Invariance to post-processing ensures that applying a data-independent procedure to the output of a DP mechanism does not compromise the privacy guarantee: if $\mathcal M$ is $\epsilon$-LDP with range $\mathcal Y$, and $f:\mathcal Y\rightarrow \mathcal Z$ is a (potentially randomized) function, then $f\circ \mathcal M$ is also $\epsilon$-LDP. Invariance to post-processing is especially important in this paper, as all of our inference procedures can be expressed as post-processing of more basic DP quantities. 

One of the most commonly used DP mechanisms is the Laplace mechanism, which adds noise to a function of interest. Importantly, the noise must be scaled proportionally to the \emph{sensitivity} of the function, which measures the worst-case magnitude by which the function's value may change between two individuals. Formally, the $\ell_1$-sensitivity of a function $f$: $\mathcal{D} \to \mathbb{R}^{k}$ is $\Delta_f = \sup_{x,y \in \mathcal{D} } ||f(x)-f(y)||_{1}$. 
\begin{proposition}[Laplace Mechanism]
\label{def:lapmech}
    Let $f: \mathcal{\mathcal{D}} \to \mathbb{R}^k$. The Laplace mechanism is defined as $M(x)=f(x)+(\nu_1,...,\nu_k)^\top,$
    where the $\nu_i$ are independent Laplace random variables,  $\nu_i \sim \mathrm{Lap}(0,\Delta f / \epsilon)$, where the density of the Laplace distribution, $\mathrm{Lap}(\mu,b)$, is given by $f(\nu | \mu,b)=\frac{1}{2b}\exp(-\frac{|\nu - \mu|}{b})$. Then $M$ satisfies $\epsilon$-LDP.
\end{proposition}
 For a binary variable (e.g., treatment assignment), a common mechanism is the randomized response.
\begin{proposition}[Randomized Response Mechanism]
\label{def:randomresponse}
    Let $Z_i\in \{0,1\}$ be a binary variable. The randomized response mechanism is defined as
    \begin{equation*}
        M(Z_i) = \begin{cases}
            Z_i   & \text{ w.p. } \frac{\exp(\epsilon)}{1+\exp(\epsilon)}\\
            1-Z_i & \text{ w.p. } \frac{1}{1+\exp(\epsilon)},
        \end{cases}
    \end{equation*}
    which satisfies $\epsilon$-LDP.
\end{proposition}

 
\section{Frequentist Approach}
\label{sec:frequentist}
\subsection{Minimax Risk Lower Bound for PATE Estimation}
In this section, we discuss frequentist estimators for $\tau$ under several privacy scenarios where variables are privatized in different manners. 
According to \citet{Duchi2018}, the minimax lower bound of the mean-squared error (MSE) for one-dimensional mean estimation is $O((N\epsilon^2)^{-1})$. In Lemma \ref{lemma:rate_PATE_estimation}, we show that this same lower bound applies to the MSE for PATE estimation as well. We let $\mathcal{M}_{\epsilon}$ denote the set of all privacy mechanisms that satisfy $\epsilon$-LDP. To ensure bounded $\ell_1$-sensitivity, we assume $Y_i(w) \in [0,1]$ for $i=1,\ldots, N$, and $\{Y_i(w)\}_{i=1}^{N}$ are drawn according to some distribution $P_w \in \mathcal{P}_w$, where $\mathcal{P}_w$ denotes a class of distributions on the sample space of potential outcomes. 
Our restriction to $Y_i(w) \in [0,1]$ is for simplicity and clarity. This follows standard practice (e.g., \citet{Lei2017}, \citet{Ferrando2022} to name a few), and our discussions can be easily generalized to bounded outcomes $Y_i(w) \in [a,b]$ with $-\infty<a<b<\infty$ using shifting and scaling factors.
We define an estimator $\hat{\tau}$ as a measurable function that maps privatized inputs to a real value, that is, $\hat{\tau}:\mathcal{X}^N \to \mathbb{R}$, where $\mathcal{X}$ generally denotes the space of privatized inputs under various privacy scenarios.   
\begin{lemma}
    \label{lemma:rate_PATE_estimation}
    For $\epsilon \in [0,1]$, there exists a constant $c$ such that
    \begin{equation}
        \label{eq:minimax_bound}
            c \min(1, (N\epsilon^2)^{-1}) \leq \inf_{M_{\epsilon} \in \mathcal{M}_{\epsilon}} \inf_{\hat{\tau}} \sup_{\substack{P_0 \in \mathcal{P}_0, \\ P_1 \in \mathcal{P}_1, \\
            p \in [0,1]}} \E[(\hat{\tau}-\tau)^2]
    \end{equation}
\end{lemma}
Lemma \ref{lemma:rate_PATE_estimation} implies that the optimal estimator of the PATE estimation problem also has the minimax lower bound $O((N\epsilon^2)^{-1})$.

\subsection{Joint Scenario with Known $p$}
\label{sec:jointprivacymech}
We first consider a scenario where all variables are jointly and separately privatized.
The observed outcomes are privatized by the Laplace mechanism. The privatized outcomes are $\Tilde{Y}_{i} = Y_{i} + \nu_{i}^{Y}$,
where $\nu_{i}^{Y} \sim \mathrm{Lap}(1/\epsilon_y)$. 
The binary treatment variable $W_i$ is privatized by the random response mechanism.
\begin{equation*}
    \Tilde{W}_i= \begin{cases}
    W_i  & \text{ w.p. } q_{\epsilon_w}=\frac{\exp(\epsilon_w)}{1+\exp(\epsilon_w)}\\
    1-W_i & \text{ w.p. } 1-q_{\epsilon_w}=\frac{1}{1+\exp(\epsilon_w)}.\\
    \end{cases}
\end{equation*}
By composition, the joint release of $(\Tilde Y_i,\Tilde W_i)_{i=1}^N$ satisfies $(\epsilon_y+\epsilon_w)$-LDP. $\Tilde{Y}_{i}$ is observed after adding noise to $Y_i$, which is either $Y_i(0)$ or $Y_i(1)$, but we cannot identify which it is through the observed variables because $W_i$ is also unobserved.

First, we propose estimators by plugging in the privatized observations into classical formulas, then derive bias correction results of the plug-in estimators.  We also provide variance estimators, enabling asymptotically accurate plug-in confidence intervals. 

We consider the following na\"ive inverse probability weighting (IPW) estimator $\Tilde{\tau}_{naive}$. This na\"ive IPW estimator is defined by plugging in privatized observations for the usual IPW estimator.
\begin{equation}
\label{eq:estimator_ipw}
     \Tilde{\tau}_{naive} = \frac{1}{N}\sum_{i=1}^{N} \bigg\{\frac{\Tilde{W}_i\Tilde{Y}_i}{\rho_1}-\frac{(1-\Tilde{W}_i)\Tilde{Y}_i}{\rho_0} \bigg\},
\end{equation}
where $\rho_w=\pr(\Tilde{W}_i=w)$ for $w=0,1$. Note that $\rho_w$ is a known marginal probability expressed by $p$ and $q_{\epsilon_w}$.
The following lemma quantifies the bias of the estimator \eqref{eq:estimator_ipw}.
\begin{lemma}
\label{lemma:bias_w_ipw}
Under Assumption \ref{asmp:common}, the estimator \eqref{eq:estimator_ipw} is biased for $\tau$. The bias is 
\begin{align*}
    \mathrm{Bias}(\Tilde{\tau}_{naive}) = \bigg ( \frac{1}{C_{p,\epsilon_w}} -1\bigg) \tau,
\end{align*}

where $ C_{p,\epsilon_w} = \frac{\rho_0\rho_1}{p(1-p)(2\qeps - 1)}$  with $q_{\epsilon_w}=\exp(\epsilon_w)/(1+\exp(\epsilon_w))$.
\end{lemma}

Let $\hat{E}_{w} = \frac{1}{\Tilde{N}_w} \sum_{i: \Tilde{W}_i=w}\Tilde{Y}_i$ and $\hat{V}_{w} =  \frac{1}{\Tilde{N}_w-1}\sum_{i: \Tilde{W}_i=w}(\Tilde{Y}_i - \hat{E}_{w})^2$,
where $\Tilde{N}_w = \sum_{i=1}^{N}\mathbbm{1}(\Tilde{W}_i=w)$ for $w=0,1$. 
In Theorem \ref{thm:joint_thorem}, we show that the estimator $C_{p,\epsilon_w} \Tilde{\tau}_{naive}$ is unbiased, consistent, and that we can construct asymptotically valid confidence intervals for PATE based on this estimator.

\begin{theorem}
\label{thm:joint_thorem}
    \begin{enumerate}
        \item (Unbiasedness \& Consistency) $C_{p,\epsilon_w} \Tilde{\tau}_{naive}$ is unbiased and consistent for $\tau$.
        \item (CLT) $\sqrt{N}(C_{p,\epsilon_w} \Tilde{\tau}_{naive} - \tau)$ converges in distribution to a mean-zero normal distribution.
        \item (Confidence Interval) The following interval is  the nominal
central confidence at the significance level $\alpha$:
\begin{align*}
    \left(  C_{p,\epsilon_w} \Tilde{\tau}_{naive}-z_{\frac{\alpha}{2}}\sqrt{\frac{\hat{\Sigma}_{naive}}{N}} ,  C_{p,\epsilon_w} \Tilde{\tau}_{naive}+z_{\frac{\alpha}{2}}\sqrt{\frac{\hat{\Sigma}_{naive}}{N}}\right),
\end{align*}
where $\hat{\Sigma}_{naive} = C_{p,\epsilon_w}^2(\frac{1}{\rho_1}\hat{V}_1+\frac{1}{\rho_0}\hat{V}_0+\frac{\rho_0}{\rho_1}\hat{E}_1^2+\frac{\rho_1}{\rho_0}\hat{E}_0^2 + 2\hat{E}_0\hat{E}_1)$.
        \item (Convergence rate) The MSE of $C_{p,\epsilon_w}\Tilde{\tau}_{naive}$ is $O((N\epsilon_y^2\epsilon_w^2)^{-1})$.
        
    \end{enumerate}
\end{theorem}
The details of the asymptotic normality and the confidence interval construction are in Supplementary Material \ref{sec:details_of_joint_thm}.
Setting $\epsilon_y=\epsilon_w=\epsilon/2$ gives MSE of $O((N\epsilon^4)^{-1})$, which matches the minimax rate \eqref{eq:minimax_bound} in terms of $N$, but not in terms of $\epsilon$. In the following sections, we see that when we use a customized privacy mechanism, rather than a na\"ive joint privatization, we can match the minimax lower bound. In the Supplementary Materials, we introduce another class of frequentist estimators: the OLS estimator, specifically under the joint scenario. We explore both the advantages and limitations of the OLS estimator in comparison to the IPW estimator within this context. 


\subsection{Custom Scenario with Known $p$}
\label{sec:customscenario}
In this section, we will tailor the privacy mechanism to the PATE estimation problem, assuming that the value $p$ is known (such as in most designed experiments). Specifically, for unit $i=1,\ldots,N$, we privatize the following variable by the Laplace mechanism: 
$A_i = \frac{W_iY_i}{p}-\frac{(1-W_i)Y_i}{1-p}$.
The sensitivity of $A$ is $\Delta_A = \max(\frac{1}{p},\frac{1}{1-p})$. The privatized value of $A$ is $\Tilde{A}_i = A_i +  \nu_{i}^{A}$,
where $\nu_{i}^{A} \sim \mathrm{Lap}(\Delta_A/\epsilon_a)$. Then, it is straightforward to show that the following IPW estimator is unbiased and consistent. 
\begin{equation}
\label{eq:estimator_custom}
     \Tilde{\tau}_{IPW} = \frac{1}{N}\sum_{i=1}^{N} \Tilde{A}_i.
\end{equation}

\begin{theorem}
\label{thm:details_of_custom_scenario}
    \begin{enumerate}
        \item (Unbiasedness \& Consistency) $\Tilde{\tau}_{IPW}$ is unbiased and consistent for $\tau$.
        \item (CLT) $\sqrt{N}(\Tilde{\tau}_{IPW} - \tau)$ converges in distribution to a mean-zero normal distribution.
        \item (Confidence Interval) The following interval is  the nominal
central confidence at the significance level $\alpha$:
\begin{align*}
    \left( \Tilde{\tau}_{IPW}-z_{\frac{\alpha}{2}}\sqrt{\frac{\hat{\Sigma}_{IPW}}{N}} ,  \Tilde{\tau}_{IPW}+z_{\frac{\alpha}{2}}\sqrt{\frac{\hat{\Sigma}_{IPW}}{N}}\right),
\end{align*}
where $\hat{\Sigma}_{IPW} = \frac{1}{N-1}\sum_{i=1}^{N}(\Tilde{A}_i - \hat{E}_A)^2$ with $\hat{E}_A=\frac{1}{N}\sum_{i=1}^{N}\Tilde{A}_i$.
        \item (Convergence rate) The MSE of 
 $\Tilde{\tau}_{IPW}$ is $O((N\epsilon_a^2)^{-1})$. 
    \end{enumerate}
\end{theorem}
The details of the asymptotic normal distribution and the confidence interval construction are provided in  Supplementary Material \ref{sec:details_of_custom_scenario}. We see in Theorem \ref{thm:details_of_custom_scenario} that the lower bound of the IPW estimator under the custom scenario matches the minimax lower bound for the locally private PATE estimation \eqref{eq:minimax_bound}, improving over the na\"ive estimator from Section \ref{sec:jointprivacymech}.

\subsection{Custom Scenario with Unknown $p$}
\label{sec:custom2}
The estimator \eqref{eq:estimator_custom} is appealing in the sense of optimality when $p$ is known, such as in randomized experiments, however, their application is restricted when $p$ is unknown. In this regard, we proceed a step further to address situations in which $p$ is inaccessible, while Assumption \ref{asmp:common} remains valid. Examples of this setting include A/B testing and clinical trials, where marketers or doctors assign treatments with an undisclosed probability (that does not depend on the covariate information).

We consider releasing the following quantities: $\mathbf{\Tilde{B}}_i = (\Tilde{B}_{i,1}, \Tilde{B}_{i,2}, \Tilde{B}_{i,3})$, where 
\begin{equation*}
     \Tilde{B}_{i,1} = W_iY_i +  \nu_{i}^{B_1}, \Tilde{B}_{i,2} = (1-W_i)Y_i +  \nu_{i}^{B_2}, \text{ and }  \Tilde{B}_{i,3} = W_i +  \nu_{i}^{B_3}, 
\end{equation*}
where $\nu_{i}^{B_j} \sim  \mathrm{Lap}(1/\epsilon_{b_j})$ for $j=1,2,3$. We also let $\Tilde{B}_{i,4} = 1 -  \Tilde{B}_{i,3}$. By composition, the joint release of $(\Tilde{B}_{i,1}, \Tilde{B}_{i,2}, \Tilde{B}_{i,3})_{i=1}^N$ satisfies $(\epsilon_{b_1}+\epsilon_{b_2}+\epsilon_{b_3})$-LDP. 

Given these privatized quantities, we construct our difference-in-means (DM) estimator as follows. 
\begin{equation}
    \label{eq:dm_custom}
    \Tilde{\tau}_{DM} = \frac{\sum_{i=1}^{N}\Tilde{B}_{i,1}}{\sum_{i=1}^{N}\Tilde{B}_{i,3}} - \frac{\sum_{i=1}^{N}\Tilde{B}_{i,2}}{\sum_{i=1}^{N}\Tilde{B}_{i,4}}.
\end{equation}

Let $\hat{E}_{B_j}=\frac{1}{N}\sum_{i=1}^{N}\Tilde{B}_{i,j}$, $\hat{V}_{B_j}=\frac{1}{N-1}\sum_{i=1}^{N}(\Tilde{B}_{i,j}-\hat{E}_{B_j})^2$ for $j=1,2,3,4$ and $\widehat{\mathrm{Cov}_{j,k}}=\frac{1}{N-1}\sum_{i=1}^{N}(\Tilde{B}_{i,j}-\hat{E}_{B_j})(\Tilde{B}_{i,k}-\hat{E}_{B_k})$ for $j \neq k$. We have the following properties for $\Tilde{\tau}_{DM}$:
\begin{theorem}
\label{thm:details_of_custom2_scenario}
    \begin{enumerate}
        \item (Consistency) $\Tilde{\tau}_{DM}$ is consistent for $\tau$.
        \item (CLT) $\sqrt{N}(\Tilde{\tau}_{DM} - \tau)$ converges in distribution to a mean-zero normal distribution.
        \item (Confidence Interval) The following interval is  the nominal
central confidence at the significance level $\alpha$:
\begin{align*}
    \left( \Tilde{\tau}_{DM}-z_{\frac{\alpha}{2}}\sqrt{\frac{\hat{\Sigma}_{DM}}{N}} ,  \Tilde{\tau}_{DM}+z_{\frac{\alpha}{2}}\sqrt{\frac{\hat{\Sigma}_{DM}}{N}}\right),
\end{align*}
where $\hat{\Sigma}_{DM} = \hat{\mathbf{e}}' \hat{\mathbf{S}} \hat{\mathbf{e}}$, with $\hat{\mathbf{e}} = ( 1/\hat{E}_{B_3}, -1/(1-\hat{E}_{B_3}), -\hat{E}_{B_1}/\hat{E}_{B_3}^2, \hat{E}_{B_2}/(1-\hat{E}_{B_3})^2)'$ and 
\begin{align*}
    \hat{\mathbf{S}} = \begin{pmatrix}
            \hat{V}_{B_1} & \widehat{\mathrm{Cov}_{1,2}} & \widehat{\mathrm{Cov}_{1,3}} & \widehat{\mathrm{Cov}_{1,4}}\\
            \widehat{\mathrm{Cov}_{2,1}} & \hat{V}_{B_2} & \widehat{\mathrm{Cov}_{2,3}} & \widehat{\mathrm{Cov}_{2,4}}\\
            \widehat{\mathrm{Cov}_{3,1}} & \widehat{\mathrm{Cov}_{3,2}} & \hat{V}_{B_3}  & \widehat{\mathrm{Cov}_{3,4}}\\
            \widehat{\mathrm{Cov}_{4,1}} & \widehat{\mathrm{Cov}_{4,2}} &  \widehat{\mathrm{Cov}_{4,3}} & \hat{V}_{B_4} \\
        \end{pmatrix}.
\end{align*}
        \item (Convergence rate) The MSE of 
 $\Tilde{\tau}_{DM}$ is $O((N(\epsilon_{b_1}^2+\epsilon_{b_2}^2+\epsilon_{b_3}^2))^{-1})$. 
    \end{enumerate}
\end{theorem}
The details of the asymptotic normal distribution and the confidence interval construction are provided in  Supplementary Material \ref{sec:details_of_custom2_scenario}.
Setting $\epsilon_{b_1}=\epsilon_{b_2}=\epsilon_{b_3} = \epsilon/3$ gives $O((N\epsilon^2)^{-1})$, which also matches the minimax lower bound of \eqref{eq:minimax_bound}, indicating the optimality of the estimator.

\subsection{Discussion on Frequentist Estimators}

The three scenarios serve different purposes. While the joint scenario permits the release of the entire synthetic dataset to analysts, it suffers from the privatization of multiple variables, thereby compromising its optimality. As discussed in the Supplementary Materials, the OLS estimator helps improve the efficiency under the joint scenario, however, the gain is limited since we must pay additional privacy budgets for covariates. In the custom scenarios, access to the complete dataset is unavailable, but the estimators attain the optimal rate of the locally private PATE estimation. While both custom estimators achieve the minimax rate, the estimator with known $p$ is able to focus its privacy budget on a single quantity, which gives improved finite sample performance; see Section \ref{sec:simulationstudies}.

When the sample size is small, or when privacy budgets are too tight, it is possible that the point estimators and interval estimators are out of support of the estimand, as the estimand is assumed to be bounded, but the observed private data are usually unbounded. Therefore, we apply additional post-processing to clamp estimators to the closest end of the support when they are out of bounds. For example, if the initial estimator is $\hat{\tau}=1.8$, then we instead set $\hat{\tau}=1.0$. However, suppose the lower and upper bounds of the estimated confidence interval are both clamped to the bounds of the support: in this case, the estimated confidence interval is not useful at all. This is a limitation of frequentist estimators arising from the trade-off between privacy and the accuracy of the analysis. This clamping processing is not necessary to achieve all the statistical properties derived in the paper. It only serves to reduce the MSE of the estimator by projecting the out-of-bound estimator to the bound. 


\section{Bayesian Approach}
\label{sec:bayesian}
\subsection{Overview of the Bayesian Methodology}
Following the Bayesian paradigm of \citet{Rubin1978}, we consider deriving the posterior distributions of the causal estimands \citep{Forastiere2016,ohnishi2021bayesian}. The key idea is the data augmentation \citep{Tanner1987} to obtain the posterior distribution of the causal estimands by imputing in turn the missing variables. The idea for estimating causal effects in the Bayesian paradigm is outlined in \citet{Rubin1978,imbens_rubin_2015}, but our unique challenges lie in the fact that neither treatment variable $W$ nor either potential outcome $Y(0),Y(1)$ is observed.

To show how Bayesian inference proceeds in our framework, consider the following joint distribution of all observed variables $\Tilde{\mathbf{O}}$ and missing variables $\mathbf{Y}(0),\mathbf{Y}(1),\mathbf{W}$: $\pr(\mathbf{Y}(0),\mathbf{Y}(1),\mathbf{W},\Tilde{\mathbf{O}}),$
where $\Tilde{\mathbf{O}}=(\Tilde{\mathbf{Y}}$, $\Tilde{\mathbf{W}}$) for the joint scenario and $\Tilde{\mathbf{O}}=\Tilde{\mathbf{A}}$ or $\Tilde{\mathbf{B}}$ for the custom scenarios.
As discussed in Section \ref{sec:regression_adjustment}, since causal effects are identifiable under randomization without covariate adjustment and incorporating covariates requires additional privacy costs for their release, we do not include covariates in our Bayesian methodologies, but the extension should be straightforward (e.g., \citet{Maceachern1999}).
In what follows, we focus on the joint scenario discussed in Section \ref{sec:jointprivacymech} to show the outline of our algorithm, but it can easily be extended to the custom scenarios, as explained in Supplementary Material.

Under the super-population perspective, the observed and missing variables are considered as a joint draw from the population distribution. Bayesian inference considers the observed values of these quantities to be realizations of random variables and the missing values to be unobserved random variables. We also assume these quantities are unit exchangeable, then de Finetti's theorem implies that there exists a vector of parameters, $\boldsymbol{\theta}$, with the prior distribution $\pr(\boldsymbol{\theta})$ such that
\begin{equation}
\label{eq:joint_dist}
\begin{split}
    \pr(\mathbf{Y}(0),\mathbf{Y}(1),\mathbf{W},\Tilde{\mathbf{Y}},\Tilde{\mathbf{W}}) = \int \pr(\boldsymbol{\theta}) \prod_{i} \pr(Y_{i}(0),Y_{i}(1),W_{i},\Tilde{Y}_{i},\Tilde{W}_{i} \mid \boldsymbol{\theta} )  d\boldsymbol{\theta}\\
    = \int \pr(\boldsymbol{\theta}) \prod_{i} \pr(W_i) \pr(\Tilde{W}_{i} \mid W_{i}) \pr(Y_{i}(0),Y_{i}(1) \mid \boldsymbol{\theta} ) \pr(\Tilde{Y}_{i} \mid Y_{i}(0),Y_{i}(1) , W_{i}) d\boldsymbol{\theta},
\end{split}
\end{equation}
which follows from the conditional independence of potential outcomes and $\Tilde{W}_i$ given $W_i$ (Lemma \ref{lemma:cond_indep} in the Supplementary Materials) and the random assignment assumption.
The distribution of $\Tilde{Y}_i$ depends not only on $Y_i(0)$ and $Y_i(1)$ but also on $W_i$ because the DP mechanism is applied to the observed outcome $Y_i= W_iY_i(1) + (1-W_i)Y_i(0)$.
Note that we know the DP mechanisms for $W$ and $Y$, that is, $\pr(\Tilde{Y}_{i} \mid Y_{i}(0),Y_{i}(1), W_{i})$  and $\pr(\Tilde{W}_{i} \mid W_{i})$ have a known functional form. 
Therefore, the modeling effort is only required for $\pr(Y_{i}(0),Y_{i}(1) \mid  \boldsymbol{\theta} )$. Under this modeling strategy, our Bayesian approach is a valid inference for PATE. Note that PATE is a function of the parameters $\boldsymbol{\theta}$, which governs the potential outcomes. Thus, it suffices to obtain the posterior draws of the posterior of the $\boldsymbol{\theta}$ for the posterior draws of PATE. 

A significant insight from \eqref{eq:joint_dist} is that the treatment assignment mechanism is \emph{not} ignorable. In conventional non-private settings, the treatment assignment model, represented as $\pr(W_i)$, is ignorable and falls out of the likelihood in Bayesian causal inference under randomization or unconfoundedness assumptions \citep{Li_Ding2023}. Yet, in a DP context, these treatment assignment variables are not directly observed. This necessitates the integration of both the treatment assignment models and their respective privacy mechanisms into our inferences.
Additionally, a nuanced but crucial point is the necessity to model both $Y_i(0)$ and $Y_i(1)$. Typically, Bayesian causal inference for PATE estimation is performed via observable data (e.g., \citet{Zigler2016,Stephens2023}). This is because the missing potential outcome eventually gets marginalized out from \eqref{eq:joint_dist} under the assumption of prior parameter independence and unconfounded assignment, thus it does not influence parameter inference. In our scenario, however, it is uncertain whether $Y_i(0)$ or $Y_i(1)$ has been privatized to yield $\Tilde{Y}_{i}$. This uncertainty calls for a data augmentation strategy for both potential outcomes.

We adopt the Dirichlet Process Mixture (DPM) to model $\pr(Y_{i}(0),Y_{i}(1) \mid W_{i}, \boldsymbol{\theta} )$ for its flexibility. The DPM is a natural Bayesian choice for density estimation problems, which fits our needs that require $\pr(Y_{i}(0),Y_{i}(1) \mid W_{i}, \boldsymbol{\theta} )$ to be estimated without assuming strong parametric forms.  
The following section and Supplementary Materials \ref{sec:bayes_detail} provide technical details of the DPM and the Gibbs sampler.

\subsection{Algorithm Outlines}
\label{sec:gibbs_outlines}
Equation \eqref{eq:joint_dist} motivates the Gibbs sampling procedures to obtain the draws from the posterior distribution of $\boldsymbol{\theta}$. This section describes the key steps of the Gibbs sampler. Each step is derived from the corresponding components of \eqref{eq:joint_dist}. For inference of DPM parameters, denoted by $\boldsymbol{\theta}=(\boldsymbol{\mu}, \boldsymbol{\Sigma}, \mathbf{u})$, we adopt an approximated blocked Gibbs sampler based on the truncation of the stick-breaking representation \citep{Ishwaran2000}, due to its simplicity. In this algorithm, we set a conservatively large upper bound, $K \leq \infty$, on the number of components that units potentially belong to. Let $C_i\in \{1,...,K\}$ denote the latent class indicators with a multinomial distribution, $C_i \sim \mathrm{Multinomial}(\mathbf{u})$ where $\mathbf{u}=(u_1,...,u_K)$ denote the weights of all components of the DPM. 
More specific details about the DPM are provided in the Supplementary Material. The algorithm proceeds as follows.
\begin{enumerate}
    \item Given $Y_{i}(0),Y_{i}(1)$, draw each $W_i$ from $\pr(W_i=1|-) =\frac{r_1}{r_0+r_1},$
    where $r_w=\pr(\Tilde{Y}_{i} \mid Y_{i}(w)) \pr(\Tilde{W}_{i} \mid W_{i}=w) \pr(W_i=w)$ for $w=0,1$.
    \item Given $\boldsymbol{\mu}$, $\boldsymbol{\Sigma}$, $\mathbf{u}$, $C_i$ and $W_i$, draw each $Y_i(0)$ and $Y_i(1)$ according to:
    \begin{equation*}
    \begin{split}
        &\pr(Y_i(W_i)|-) \propto \pr(Y_i(W_i) \mid \mu_{W_i}^{C_i}, \Sigma_{W_i}^{C_i})\pr(\Tilde{Y}_{i} \mid Y_{i}(W_i)) \\
        &\pr(Y_i(1-W_i)|-) \propto \pr(Y_i(1-W_i) \mid \mu_{1-W_i}^{C_i}, \Sigma_{1-W_i}^{C_i}).
    \end{split}
    \end{equation*}
    \item Given $\boldsymbol{\mu}$, $\boldsymbol{\Sigma}$, $\mathbf{u}$, $Y_i(0)$ and $Y_i(1)$, draw each $C_i$ from
    \begin{equation*}
        \pr(C_i=k|-) \propto u_k \pr(Y_i(0) \mid \mu_{0}^{k}, \Sigma_{0}^{k}) \pr(Y_i(1) \mid \mu_{1}^{k}, \Sigma_{1}^{k}).
    \end{equation*}
    \item Let $u_K'=1$. Given $\alpha$, $\mathbf{C}$, draw $u_k'$ for $k \in \{1,...,K-1\}$ from
    \begin{equation*}
        \pr(u_k'|-) \propto \text{Beta} \bigg(1+ \sum_{i:C_i=k}1, \alpha+ \sum_{i:C_i>k}1 \bigg).
    \end{equation*}
    Then, update $u_k= u_k'\prod_{j<k}(1-u_j')$.
    \item Given $\mathbf{C}$ and $\mathbf{u}'$, draw $\alpha$ from 
    \begin{equation*}
        \pr(\alpha | -) \propto \pr(\alpha) \prod_{k=1}^{K} f \bigg(u_k' \bigg| 1+ \sum_{i:C_i=k}1, \alpha+ \sum_{i:C_i>k}1 \bigg),
    \end{equation*}
    where $f$ is the pdf of $u_k'$, the beta distribution. The standard Metropolis-Hastings algorithm is used for this step. 
    \item Given $\mathbf{Y}(0)$, $\mathbf{Y}(1)$ and $\mathbf{C}$, draw $\boldsymbol{\mu}$ and $\boldsymbol{\Sigma}$ from
    \begin{equation*}
    \begin{split}
        \pr(\mu_0^k, \Sigma_0^k|-) &\propto H(\mu_0^k, \mu_1^k, \Sigma_0^k, \Sigma_1^k) \prod_{i:C_i=k}\pr(Y_i(0),Y_i(1) \mid  \mu_0^k, \mu_1^k, \Sigma_0^k, \Sigma_1^k).
    \end{split}
    \end{equation*}
\end{enumerate}
\begin{remark}
     The key steps of this algorithm are 1 and 2, which correspond to the data augmentation steps, imputing the latent variables $Y_{i}(0),Y_{i}(1)$ and $W_i$.  In Step 1, the probability $\pr(\Tilde{Y}_{i} \mid Y_{i}(w))$ for $w=0,1$ indicates that $\Tilde{Y}_{i}$ is observed via privatizing the potential outcome $Y_{i}(w)$, which would have been observed if we observed $W_i=w$. In step 2, given $W_i$, the corresponding potential outcome $Y_i(W_i)$ is considered to be privatized, but the other missing potential outcome $Y_i(1-W_i)$ should not be associated with the observed $\Tilde{Y}_{i}$ within the iteration. Therefore, the posterior distribution of $Y_i(W_i)$ cannot be obtained in a closed form as it is weighted by the privacy mechanism $\pr(\Tilde{Y}_{i} \mid Y_{i}(W_i))$, whereas the missing potential outcomes $Y_i(1-W_i)$ are just generated from the outcome model $\pr(Y_i(1-W_i) \mid \boldsymbol{\theta})$. We adopt the privacy-aware Metropolis-within-Gibbs algorithm proposed in \citet{Nianqiao_Jordan_2022} for the posterior draws of $Y_i(W_i)$. They proposed a generic data augmentation approach of updating confidential data that exploits the privacy guarantee of the mechanism to ensure efficiency. Their algorithm has guarantees on mixing performance, indicating that the acceptance probability is lower bounded by $\exp(- \epsilon_y)$. Another advantage of their approach is that we may utilize the outcome model to sample a proposal value from $\pr(Y_i(W_i)|\theta)$ at the current value of $\theta$, rather than specifying a custom proposal distribution and step size for the Metropolis-Hastings step. Finally, Steps 3--6 updates all the parameters of the DPM that govern the potential outcomes, using standard techniques; see Section \ref{sec:bayes_detail} of the Supplementary Materials for details of the DPM, full details of the algorithm and the extension of the algorithm to the custom scenarios, which requires slight modifications to Steps 1 and 2.

\end{remark}

\section{Simulation Studies}
\label{sec:simulationstudies}
We evaluate the frequentist properties of our methodologies for various privacy budgets. The evaluation metrics that we consider are bias and mean square error (MSE) in estimating a causal estimand, coverage of an interval estimator for a causal estimand, and the interval length. 
Bias, MSE  and coverage are generally defined as $\sum_{m=1}^M \left ( \tau - \hat{\tau}_m \right )/M$, $\sum_{m=1}^M \left ( \tau - \hat{\tau}_m \right )^2/M$ and $\sum_{m=1}^M \mathbbm{1} \left ( \hat{\tau}_m^{l} \leq \tau \leq \hat{\tau}_m^{u} \right )/M$ respectively, where $M$ denotes the number of simulated datasets, $\tau$ denotes the true causal estimand,  $\hat{\tau}_m$, $\hat{\tau}_m^{l}$ and $\hat{\tau}_m^{u}$ denote the estimate of the causal estimand, $95\%$ lower and upper end of the interval estimator of the causal estimand using dataset $m = 1, \ldots, M$.  Our summary of the interval length is the mean of the lengths of the intervals computed from $M$ simulated datasets. For our Bayesian method, the point estimator is the mean of the posterior distribution of a causal estimand, and the interval estimator is the $95\%$ central credible interval.
We ran the MCMC algorithm for $100,000$ iterations using a burn-in of $50,000$. The iteration numbers were chosen after experimentation to deliver stable results over multiple runs. 

\subsection{Data-generating Mechanisms}
\label{sec:simulationsetting}
For our simulations, we consider a Bernoulli randomized experiment with treatment assignment and covariates for unit $i$ generated according to:
\begin{equation*}
    W_{i} \sim \mathrm{Bernoulli}(0.5), X_{i,1} \sim \mathrm{Uniform}(0,1), X_{i,2} \sim \mathrm{Beta}(2,5), X_{i,3} \sim \mathrm{Bernoulli}(0.7).
\end{equation*}
To generate potential outcomes, we adopt the Beta regression \cite{Cribari-Neto_2004}: $Y_i(w) \sim \mathrm{Beta}(\mu_i(w)\phi, (1-\mu_i(w))\phi)$,
where $\mu_i(w)$ and $\phi$ are a location parameter and scale parameter respectively with $\mu_i(w) = \mathrm{expit}(1.0-0.8X_1+0.5X_2-2.0X_3+0.5w)$ and $\phi=50$. We consider $X_{i,d}$ to generate $Y_i$ but do not release the privatized $\Tilde{X}_{i,d}$. This model is beneficial for our simulations because the generated data automatically satisfy the following sensitivity: $\Delta_{Y}=1$.
 Then, we obtain the private data $ \Tilde{Y}_i$, $ \Tilde{W}_i$, $ \Tilde{A}_i$, $ \mathbf{\Tilde{B}}_i$ by applying the corresponding privacy mechanisms. 
The actual value of PATE can be obtained in a closed form, which is necessary to calculate bias, MSE, and coverage. The details are provided in the Supplementary Materials.

\subsection{Results}\label{sec:results}
Table \ref{tab:simulation_freq_N10000_IPW} presents the performance evaluation of our estimators under different scenarios for $N=10000$ with various privacy budgets for  $\epsilon_{tot}$. We let $\epsilon_{tot} = \epsilon_a = \epsilon_y + \epsilon_w= \epsilon_{b_1} + \epsilon_{b_2}+ \epsilon_{b_3}$, where $\epsilon_{y}=\epsilon_{w}$ and  $\epsilon_{b1}=\epsilon_{b2}=\epsilon_{b3}$.  All scenarios achieve about $95\%$ coverage, except for the custom scenario (DM) of $\epsilon_{\mathrm{tot}}=0.1,.03$, which has some over-coverage. This may be because the estimator for the asymptotic variance has a non-negligible estimation error with the finite samples.
The simulations in this section rely on the results of Section \ref{sec:jointprivacymech}, \ref{sec:customscenario}, and \ref{sec:custom2} to build confidence intervals. The fact that the intervals have correct $95\%$ coverage indicates that the estimators 1) are in fact asymptotically normal, 2) are asymptotically unbiased, and 3) have the stated asymptotic variance.
For bias and MSE, we observe smaller bias and MSE for larger privacy budgets. 
The custom scenario (IPW) yields lower MSE than the joint scenario, which is also consistent with the discussion of the optimality in Section \ref{sec:jointprivacymech}, \ref{sec:customscenario}, and \ref{sec:custom2}, but the difference becomes negligible as $\epsilon_{\mathrm{tot}}$ increases. 

When we have a tight privacy budget of $\epsilon_{\mathrm{tot}}=0.1,0.3$, the length of the confidence intervals of the joint scenario are nearly $2$, which is almost non-informative about the estimand. With strict budget constraints and a small sample size, the analysis results may tell us little about the estimands, even though their consistency and confidence intervals are statistically valid. This is an inevitable trade-off between privacy protection and the accuracy of the results. Custom (IPW) has the best finite sample performance, offering informative intervals and small bias and MSE for all privacy budgets.

Table \ref{tab:simulation_bayes_jointcustom} compares our Bayesian methodology under the three scenarios. We see that the Bayes estimator yields well-calibrated coverage probabilities and smaller MSE and bias for most cases. The differences in MSE between frequentist estimators and Bayesian estimators become negligible as $\epsilon_{\mathrm{tot}}$ gets large ($\epsilon_{\mathrm{tot}}=3.0,10.0$).  When the privacy budget is tight, the Bayesian methodology outperforms the frequentist approach in all metrics. Specifically, the interval length of the Bayes estimator for $\epsilon_{\mathrm{tot}}=0.1$ is around $0.35$ for all scenarios, which is informative enough about the estimands. 
In the Supplementary Materials, we provide additional simulation studies for smaller sample sizes, as well as those for the OLS estimator under the joint scenario.

\begin{table*}
	\centering
	\caption{Evaluation metrics for IPW estimator under different privacy scenarios ($N=10000, N_{sim}=2000$). $N_{sim}$ denotes the number of simulations. $\epsilon_{\mathrm{tot}}$ denotes the total privacy budget. ``Custom (IPW)'' and ``Custom (DM)'' columns are scenarios in Section \ref{sec:customscenario} and \ref{sec:custom2} respectively.}
	\begin{adjustbox}{width=16.5cm}
		\begin{tabular}{rrrrrrrrrrrrr}
			\toprule
			  &\multicolumn{3}{c}{Coverage}   & \multicolumn{3}{c}{Bias}   & \multicolumn{3}{c}{MSE}   & \multicolumn{3}{c}{Interval Width} \\
			\cmidrule(lr){2-4} \cmidrule(lr){5-7} \cmidrule(lr){8-10} \cmidrule(lr){11-13} 
			 $\epsilon_{\mathrm{tot}}$ &  Joint & Custom (IPW) & Custom (DM) & Joint & Custom (IPW) & Custom (DM) & Joint & Custom (IPW) & Custom (DM) & Joint & Custom (IPW) & Custom (DM) \\ \hline
			0.1 & $94.55\%$ & $94.95\%$ & $99.8\%$ & $0.9025$ & $-0.1873$ & $0.9025$ & $0.9872$ & $0.0803$ & $0.7608$ & $1.889$ & $1.091$ & $1.988$ \\
0.3 & $94.7\%$ & $94.1\%$ & $98.05\%$ & $0.9025$ & $-0.0221$ & $-0.4396$ & $0.7875$ & $0.0091$ & $0.2518$ & $1.882$ & $0.371$ & $1.655$ \\
1.0 & $94.65\%$ & $94.6\%$ & $95.6\%$ & $-0.2171$ & $-0.0086$ & $-0.1498$ & $0.0568$ & $0.0009$ & $0.0201$ & $0.915$ & $0.117$ & $0.553$ \\
3.0 & $95.3\%$ & $95.0\%$ & $95.3\%$ & $-0.033$ & $-0.0078$ & $0.0076$ & $0.0011$ & $0.0002$ & $0.0022$ & $0.13$ & $0.052$ & $0.182$ \\
10.0 & $94.9\%$ & $94.95\%$ & $94.4\%$ & $0.0$ & $0.003$ & $0.0012$ & $0.0001$ & $0.0001$ & $0.0002$ & $0.043$ & $0.038$ & $0.057$ \\
			\bottomrule
		\end{tabular}
	\end{adjustbox}
	\label{tab:simulation_freq_N10000_IPW}
\end{table*}

\begin{table*}
	\centering
	\caption{Evaluation metrics of Bayesian estimators for $N=10000, N_{sim}=1000$.}
	\begin{adjustbox}{width=16.5cm}
		\begin{tabular}{rrrrrrrrrrrrr}
			\toprule
			  &\multicolumn{3}{c}{Coverage}   & \multicolumn{3}{c}{Bias}   & \multicolumn{3}{c}{MSE}   & \multicolumn{3}{c}{Interval Width} \\
			\cmidrule(lr){2-4} \cmidrule(lr){5-7} \cmidrule(lr){8-10} \cmidrule(lr){11-13} 
			 $\epsilon_{\mathrm{tot}}$ &  Joint & Custom (IPW) & Custom (DM) & Joint & Custom (IPW) & Custom (DM) & Joint & Custom (IPW) & Custom (DM) & Joint & Custom (IPW) & Custom (DM) \\ \hline
			0.1 & $96.4\%$ & $93.8\%$ & $96.3\%$ & $-0.0949$ & $-0.0772$ & $-0.0951$ & $0.0099$ & $0.0079$ & $0.0099$ & $0.34$ & $0.319$ & $0.341$ \\
0.3 & $96.9\%$ & $94.5\%$ & $94.6\%$ & $-0.0953$ & $-0.036$ & $-0.0897$ & $0.0099$ & $0.0034$ & $0.0094$ & $0.342$ & $0.22$ & $0.334$ \\
1.0 & $93.4\%$ & $93.8\%$ & $92.2\%$ & $-0.0691$ & $-0.0069$ & $-0.0511$ & $0.0077$ & $0.0006$ & $0.0055$ & $0.32$ & $0.096$ & $0.263$ \\
3.0 & $93.2\%$ & $92.2\%$ & $94.2\%$ & $-0.0081$ & $-0.0063$ & $-0.0098$ & $0.0006$ & $0.0002$ & $0.001$ & $0.093$ & $0.045$ & $0.117$ \\
10.0 & $95.0\%$ & $93.5\%$ & $92.3\%$ & $-0.0023$ & $-0.0027$ & $-0.0045$ & $0.0$ & $0.0$ & $0.0001$ & $0.026$ & $0.022$ & $0.036$ \\
			\bottomrule
		\end{tabular}
	\end{adjustbox}
 \label{tab:simulation_bayes_jointcustom}
\end{table*}
\vspace{-10pt}
\section{Real Data Analysis}
\label{sec:empirical_analysis}
We applied our methodology to a real-world causal inference task. We analyzed a randomized experiment that examined the impact of a cash transfer program on students' attendance rates \citep{Barrera-Osorio2011}. Conducted at San Cristobal in Colombia, the study recruited households with one to five school children, randomly assigning children to either participate in the cash transfer program or not with probability $p=0.628$. The number of recruited students is $N=5240$. With known treatment assignment, we assessed the treatment effect of the program on the attendance rate of the students, with eligible students receiving cash subsidies if they attended school at least $80\%$ of the time in a given month.

We utilized the privatization techniques as outlined in Section \ref{sec:frequentist}, setting $\epsilon_{\mathrm{tot}}$ to values of $0.1$, $0.3$, $1.0$, $3.0$, and $10.0$. Our methodologies were then benchmarked against non-private baseline methods, which offer target values for our private estimates. For the non-private frequentist baseline, we employed the standard IPW estimator.

Table \ref{tab:empirical_analysis} presents point mean estimators alongside the lower ($2.5\%$) and upper ($97.5\%$) bounds for interval estimators across each methodology. For the interval estimators, we used central confidence intervals for the frequentist approach and credible intervals for the Bayesian approach. Both frequentist and Bayesian non-private interval estimators highlighted a positive interval, indicating a significant effect. The point estimates showed a $0.6\%$ increase in the frequentist non-private approach and a more modest $0.5\%$ increase in the Bayesian approach. Given these results, our expectation for the private methodologies is, at best, to approximate the non-private values, since better inferences are unlikely with privatized data. Note that as the experimental data is fixed, the only randomness in this study is the privacy mechanisms.

The point estimates for both frequentist and Bayesian methodologies are similar to their non-private results when $\epsilon_{\mathrm{tot}} \geq 3.0$. In particular, we observe that the Custom (IPW) scenario results in the narrowest confidence intervals. In the joint and custom (DM) scenarios, the frequentist estimators deviated more from the non-private one, showing larger intervals. The frequentist methodologies yield non-informative intervals when the privacy budget is tightest $\epsilon_{\mathrm{tot}} = 0.1$. The Bayesian methodology demonstrated strong performance across all scenarios. These observations align with our simulation studies, further validating the efficacy of our methodologies.

\begin{table*}
	\label{tab:empirical_analysis}
		\centering
		\caption{Empirical analysis evaluating privatized cash transfer programs in Colombia. In the "Non-private" columns, "Freq" represents the standard IPW estimator, while "Bayes" represents the standard Dirichlet process mixture models for non-private data.  }
		\begin{adjustbox}{width=16.cm}
			\begin{tabular}{rrrrrrrrrrrrrrrrrrrrrrrrr} 
				  & \multicolumn{6}{c}{Non-private}   & \multicolumn{18}{c}{Private}  \\
				  \cmidrule(lr){2-7} \cmidrule(lr){8-25}
				 &\multicolumn{3}{c}{} &\multicolumn{3}{c}{}   & \multicolumn{6}{c}{Joint}  & \multicolumn{6}{c}{Custom (IPW)} & \multicolumn{6}{c}{Custom (DM)} \\
					 \cmidrule(lr){8-13} \cmidrule(lr){14-19} \cmidrule(lr){20-25} 
					& \multicolumn{3}{c}{Freq} & \multicolumn{3}{c}{Bayes}   & \multicolumn{3}{c}{Freq} & \multicolumn{3}{c}{Bayes} & \multicolumn{3}{c}{Freq} & \multicolumn{3}{c}{Bayes}& \multicolumn{3}{c}{Freq} & \multicolumn{3}{c}{Bayes}\\
					\cmidrule(lr){2-4}\cmidrule(lr){5-7} \cmidrule(lr){8-10} \cmidrule(lr){11-13} \cmidrule(lr){14-16} \cmidrule(lr){17-19} \cmidrule(lr){20-22} \cmidrule(lr){23-25} 
				$\epsilon_{\mathrm{tot}}$ & Mean & 2.5\% & 97.5\% & Mean & 2.5\% & 97.5\% & Mean & 2.5\% & 97.5\% &Mean & 2.5\% & 97.5\% &Mean & 2.5\% & 97.5\% &Mean & 2.5\% & 97.5\% &Mean & 2.5\% & 97.5\% &Mean & 2.5\% & 97.5\% \\
				\hline
				0.1 & 0.006 & -0.042 & 0.054 & 0.005 & 0.001 & 0.008 & 1.0    & -1.0   & 1.0    & 0.011 & -0.178 & 0.145 & 0.019 & -1.0       & 1.0   & 0.072  & -0.135 & 0.244 & 1.0  & -1.0   & 1.0   & 0.032 & -0.137 &  0.193 \\
				0.3 & 0.006 & -0.042 & 0.054 & 0.005 & 0.001 & 0.008 & -1.0 & -1.0   & 1.0      & 0.049 & -0.082 & 0.190 & 0.010 & -0.367     & 0.386 & 0.160  & -0.038 & 0.389 & -0.581  & -1.0   & 1.0   & 0.049 & -0.148 & 0.238 \\
				1.0 & 0.006 & -0.042 & 0.054 & 0.005 & 0.001 & 0.008 & -0.169  & -0.898 & 0.559 & 0.041 & -0.022 & 0.111 & 0.006 & -0.118     & 0.131 & 0.073  & -0.018 & 0.137 & 0.131  & -0.546 & 0.809 & 0.054 & -0.124 & 0.169 \\
				3.0 & 0.006 & -0.042 & 0.054 & 0.005 & 0.001 & 0.008 & 0.008  & -0.116 & 0.131  & 0.018 & -0.007 & 0.044 & 0.005 & -0.061     & 0.072 & -0.002 & -0.018 & 0.015 & -0.038  & -0.248 & 0.170 & 0.048 & 0.004  & 0.098 \\
			   10.0 & 0.006 & -0.042 & 0.054 & 0.005 & 0.001 & 0.008 & 0.006  & -0.051  & 0.064 & 0.009 & 0.0    & 0.018 & 0.008 & -0.048     & 0.064 & -0.002 & -0.009 & 0.006 & -0.006  & -0.066 & 0.054 & 0.015 & 0.002  & 0.027 \\
		\bottomrule
			\end{tabular}
		\end{adjustbox}
	\end{table*}

\section{Concluding Remarks}
\label{sec:conclusion}
In this article we proposed causal inferential methodologies to analyze differential private data under the Rubin Causal Model. We considered three distinct local privacy scenarios that have practical relevance: 1) jointly privatized variables with known $p$, 2) custom privatized variables with known $p$, and 3) custom privatized variables with unknown $p$. 
We showed that a na\"ive debiased estimator in the first scenario results in poor MSE compared to the minimax lower bound. In contrast, we show that by using customized privacy mechanisms, we can achieve the lower bound and obtain minimax optimal inference. We also presented a Bayesian methodology and its sampling algorithm as an alternative to the frequentist methodologies. 
We emphasize that despite the simplicity of the Laplace and randomized response mechanisms we employ, our customized estimators attain the minimax lower bound, thereby ensuring optimality across any privacy mechanisms. Thus, the mechanism choice is of lesser concern. Additionally, our analyses can readily be extended to other mechanisms that add independent noise with a zero mean and known variance. Our Bayesian algorithm works effectively across a broad spectrum of privacy mechanisms if the privacy mechanism has a known likelihood.
Finally, we validated the performance of our estimators via simulation studies and empirical analyses using real-world data.

A direction for future research is to develop an analytical framework for unbounded variables. Our framework is restricted to bounded variables due to considerations of the sensitivity of DP mechanisms. 

Furthermore, the finite-sample performance of our estimators may be improved by more carefully designing the noise adding mechanisms; one may investigate using truncated-uniform-Laplace \citep{awan2018differentially}, $K$-norm mechanisms \citep{hardt2010geometry,awan2021structure}, or the minimax optimal noise mechanism for multivariate mean estimation \citep{Duchi2018}. 

Finally, another direction of future work would be to develop methodologies for the PATE estimation in observational studies.

\begin{singlespace}
\bibliographystyle{Chicago}
\bibliography{localCI}
\end{singlespace}

\newpage
\appendix
\onecolumn

\section{Details of Theorems and Proofs in Section \ref{sec:frequentist}}

\subsection{Conditional independence of $\{Y_{i}(0),Y_{i}(1)\}$ and $\Tilde{W}_i$ given $W_i$}
We first state a subtle yet important lemma that we will use to prove subsequent theorems.
\begin{lemma}
\label{lemma:cond_indep}
     The potential outcomes are conditionally independent of the privatized treatment assignments given the actual treatment assignment:
        \begin{equation*}
            \{Y_{i}(0),Y_{i}(1)\} \indep \Tilde{W}_i \mid W_i.
        \end{equation*}
\end{lemma}
This result holds because the DP mechanism flips the given treatment independently. This result is subtle, but important because it plays a crucial role in proving the upcoming theorems. 

\subsection{Proof of Lemma \ref{lemma:rate_PATE_estimation}}
\begin{proof}
    We first acknowledge that 
\begin{equation}
\label{eq:PATE_to_mean_estimation}
     \sup_{\substack{P_0 = \delta(0), \\ P_1 \in \mathcal{P}_1, \\
            p =1}} \E[(\hat{\tau}-\tau)^2] = \sup_{ P_1 \in \mathcal{P}_1} \E[(\hat{\tau}-\mu_1)^2],
\end{equation}
where $\delta(0)$ denotes a point mass at $0$ and $\mu_1= \E[Y_i(1)]$. Equation \eqref{eq:PATE_to_mean_estimation} is equivalent to the one-dimensional mean estimation problem in \citet[Corollary 1]{Duchi2018}. Therefore, by \citet{Duchi2018}, there exists some constant $c_l$ such that
\begin{align*}
      c_l \min(1,(N\epsilon^2)^{-1})\leq \sup_{ P_1 \in \mathcal{P}_1} \E[(\hat{\tau}-\mu_1)^2],
\end{align*}
Finally, we note that 
\begin{align*}
    \inf_{M_{\epsilon} \in \mathcal{M}_{\epsilon}} \inf_{\hat{\tau}} \sup_{\substack{P_0 = \delta(0), \\ P_1 \in \mathcal{P}_1, \\
            p =1}} \E[(\hat{\tau}-\tau)^2]  \leq \inf_{M_{\epsilon} \in \mathcal{M}_{\epsilon}} \inf_{\hat{\tau}} \sup_{\substack{P_0 \in \mathcal{P}_0, \\ P_1 \in \mathcal{P}_1, \\
            p \in [0,1]}} \E[(\hat{\tau}-\tau)^2],
\end{align*}
where the inequality holds as the right side is taking supremum over a larger set.
Putting everything together, we prove our claim.
\end{proof}

\subsection{Proof of Lemma \ref{lemma:bias_w_ipw}}
\label{proof:bias_w_ipw}
\begin{proof}
Let $\pbar=1-p$ and $\qbareps=1-\qeps$.
The weak law of large numbers implies
\begin{align*}
     \frac{1}{N} \sum_{i=1}^{N}\Tilde{W}_i\Tilde{Y}_i &\convp \E[\Tilde{W}_i\Tilde{Y}_i] \\
     & = \E[\E[\Tilde{W}_i\Tilde{Y}_i \mid W_i]]\\
     & = \E[\E[\Tilde{W_i} \mid W_i]\E[Y_i \mid W_i]]\\
     & = \E[\pr(\Tilde{W}_i=1 \mid W_i)\E[\Tilde{Y}_i \mid W_i]]\\
     & = \E[\pr(\Tilde{W}_i=1 \mid W_i)\E[Y_i \mid W_i]]\\
     & = p\big( \pr(\Tilde{W}_i=1 \mid W_i=1) \E[Y_i(1)]\big) + \pbar\big( \pr(\Tilde{W}_i=1 \mid W_i=0) \E[Y_i(0)]\big)\\
     & = p \qeps \mu_1 + \pbar\qbareps\mu_0,
\end{align*}
where the second line follows from the law of total expectation and the third line follows from Lemma \ref{lemma:cond_indep}. Similarly, we have 

\begin{align*}
    \frac{1}{N} \sum_{i=1}^{N}(1-\Tilde{W}_i)\Tilde{Y}_i &\convp p\qbareps\mu_1 + \pbar\qeps \mu_0.
\end{align*}
Therefore, we see that 
\begin{align*}
    \Tilde{\tau}_{naive} \convp \frac{1}{ C_{p,\epsilon_w}}\tau,
\end{align*}
and, since $ C_{p,\epsilon_w}$ is a constant, we have
\begin{align*}
    C_{p,\epsilon_w}\Tilde{\tau}_{naive} \convp \tau.
\end{align*}
\end{proof}

\subsection{Details of Theorem \ref{thm:joint_thorem}}
\label{sec:details_of_joint_thm}
We provide the following central limit theorem.
\begin{theorem}
\label{thm:clt_ipw}
The estimator $C_{p,\epsilon_w} \Tilde{\tau}_{naive}$ is unbiased and consistent for $\tau$. Furthermore, $\sqrt{N}(C_{p,\epsilon_w} \Tilde{\tau}_{naive} - \tau)$ converges in distribution to 
    \begin{equation}\label{eq:CLT_IPW}
    \mathrm{N} \bigg(0, C_{p,\epsilon_w}^2\bigg(\frac{1}{\rho_1}V_1+\frac{1}{\rho_0}V_0+\frac{\rho_0}{\rho_1}E_1^2+\frac{\rho_1}{\rho_0}E_0^2 + 2E_0E_1\bigg) \bigg),
    \end{equation}
    where, for $w=0,1$,
    \begin{align*}
    V_w = \Var( \Tilde{Y}_i | \Tilde{W}_i=w)
        =& \pr(W_i=0 | \Tilde{W}_i=w)\Var[Y_i(0)] + \pr(W_i=1 | \Tilde{W}_i=w)\Var[Y_i(1)] \\
        &+  \pr(W_i=0 | \Tilde{W}_i=w) \pr(W_i=1 | \Tilde{W}_i=w)\tau^2 + \frac{2}{\epsilon_y^2},
    \end{align*}
    and $E_w = \E( \Tilde{Y}_i | \Tilde{W}_i=w) = \pr(W_i=0 | \Tilde{W}_i=w)\E[Y_i(0)] + \pr(W_i=1 | \Tilde{W}_i=w)\E[Y_i(1)]$.
\end{theorem}
\begin{proof}
    Consistency is proven in Section \ref{proof:bias_w_ipw}.
    \begin{align*}
    \label{eq:estimator_ipw}
         C_{p,\epsilon_w}\Tilde{\tau}_{naive} = \frac{C_{p,\epsilon_w}}{N}\sum_{i=1}^{N} \bigg\{\frac{\Tilde{W}_i\Tilde{Y}_i}{\rho_1}-\frac{(1-\Tilde{W}_i)\Tilde{Y}_i}{\rho_0} \bigg\} = \frac{C_{p,\epsilon_w}}{N}\sum_{i=1}^{N} \Tilde{\tau}_{i}.
    \end{align*}
    Note that $\Tilde{\tau}_{i}$ is i.i.d. for $i=1,\ldots,N$, $\E[C_{p,\epsilon_w}\Tilde{\tau}_{i}] =\tau$, and the second moment is bounded due to the sensitivity of $Y$. Thus, it is sufficient to derive the variance of $\Tilde{\tau}_{i}$ as $\Var[C_{p,\epsilon_w}\Tilde{\tau}_{i}] = C_{p,\epsilon_w}^2  \Var[\Tilde{\tau}_{i}]$.
    \begin{align*}
        \Var[\Tilde{\tau}_{i}] = \frac{1}{\rho_1^2}\Var[\Tilde{W}_i\Tilde{Y}_i] + \frac{1}{\rho_0^2}\Var[(1-\Tilde{W}_i)\Tilde{Y}_i] - \frac{2}{\rho_0\rho_1}\Cov[\Tilde{W}_i\Tilde{Y}_i, (1-\Tilde{W}_i)\Tilde{Y}_i].
    \end{align*}
    Then, 
    \begin{align*}
        \Var[\Tilde{W}_i\Tilde{Y}_i] &= \E[\Var[\Tilde{W}_i\Tilde{Y}_i \mid \Tilde{W}_i]] + \Var[\E[\Tilde{W}_i\Tilde{Y}_i \mid \Tilde{W}_i]] \\
        &= \E[\Tilde{W}_i^2\Var[\Tilde{Y}_i \mid \Tilde{W}_i]] + \Var[\Tilde{W}_i\E[\Tilde{Y}_i \mid \Tilde{W}_i]] \\
        &= p(\Tilde{W}_i=1)\Var[\Tilde{Y}_i \mid \Tilde{W}_i] + p(\Tilde{W}_i=1)p(\Tilde{W}_i=0)\E[\Tilde{Y}_i \mid \Tilde{W}_i]^2 \\
        &= \rho_1\Var[\Tilde{Y}_i \mid \Tilde{W}_i] + \rho_0\rho_1\E[\Tilde{Y}_i \mid \Tilde{W}_i]^2 \\
        &= \rho_1V_1 + \rho_0\rho_1E_1^2. \\
    \end{align*}
    Similarly, we have $\Var[(1-\Tilde{W}_i)\Tilde{Y}_i] = \rho_0V_0 + \rho_0\rho_1E_0^2$.
    The covariance is given by 
    \begin{align*}
        \Cov[\Tilde{W}_i\Tilde{Y}_i, (1-\Tilde{W}_i)\Tilde{Y}_i] &= - \E[\Tilde{W}_i\Tilde{Y}_i]\E[(1-\Tilde{W}_i)\Tilde{Y}_i] \\
        &= -\E[\Tilde{W}_i\E[\Tilde{Y}_i \mid \Tilde{W}_i]]\E[(1-\Tilde{W}_i)\E[\Tilde{Y}_i \mid \Tilde{W}_i]] \\
        &= -p(\Tilde{W}_i=1)\E[\Tilde{Y}_i \mid \Tilde{W}_i=1] p(\Tilde{W}_i=0)\E[\Tilde{Y}_i \mid \Tilde{W}_i=0] \\
        &= -\rho_0\rho_1E_0E_1.
    \end{align*}
    Putting all together, we prove the central limit theorem in Theorem \ref{thm:clt_ipw} and hence Theorem \ref{thm:joint_thorem}.

    Next, we consider the decompositions of $E_w$ and $V_w$. We have
    \begin{align*}
        E_w &= \E[\hat{Y}_i \mid \hat{W}_i=w ] = \E[Y_i \mid \hat{W}_i=w ] = \pr(W_i=0 | \Tilde{W}_i=w)\E[Y_i(0)] + \pr(W_i=1 | \Tilde{W}_i=w)\E[Y_i(1)],
    \end{align*}
    which follows from Lemma \ref{lemma:cond_indep}. 
    By the law of total variance and SUTVA,
    \begin{align*}
        \Var[Y_i \mid \Tilde{W}_i=1] = \sum_{w=0}^{1}\Var[Y_i \mid \Tilde{W}_i=1, W_i=w]\pr(W_i=w \mid \Tilde{W}_i=1) + \Var[\E[Y_i \mid \Tilde{W}_i=1, W_i=w]].
    \end{align*}
    The first term simplifies to
    \begin{align*}
        \sum_{w=0}^{1}\Var[Y_i \mid \Tilde{W}_i=1, W_i=w]\pr(W_i=w \mid \Tilde{W}_i=1) = \frac{\pbar\qbareps}{\pq}\Var[Y_i(0)]+\frac{p\qeps}{\pq}\Var[Y_i(1)].
    \end{align*}
    The second term simplifies to
    \begin{align*}
        &\Var[\E[Y_i \mid \Tilde{W}_i=1, W_i=w]] \\
        &= \E \big[(\E[Y_i \mid \Tilde{W}_i=1, W_i=w] - \E[(\E[Y_i \mid \Tilde{W}_i=1, W_i=w] ])^2 \mid \Tilde{W}_i=1 \big] \\
        &= \sum_{w=0}^{1}\bigg( \E[Y_i(w)] - \frac{\pbar\qbareps}{\pq}\E[Y_i(0)]- \frac{p\qeps}{\pq}\E[Y_i(1)]\bigg)^2 \pr(W_i=w \mid \Tilde{W}_i=1)\\
        &= \frac{p\qeps\pbar\qbareps}{(\pq)^2}\tau^2.
    \end{align*}
    Therefore, we have 
    \begin{align*}
        V_1:=\Var[Y_i \mid \Tilde{W}_i=1] = \frac{\pbar\qbareps}{\pq}\Var[Y_i(0)]+\frac{p\qeps}{\pq}\Var[Y_i(1)] + \frac{p\qeps\pbar\qbareps}{(\pq)^2}\tau^2.
    \end{align*}
    Similarly, we have 
    \begin{align*}
        V_0:=\Var[Y_i \mid \Tilde{W}_i=0] = \frac{\pbar\qeps}{\pqbar}\Var[Y_i(0)]+\frac{p\qbareps}{\pqbar}\Var[Y_i(1)] + \frac{p\qeps\pbar\qbareps}{(\pqbar)^2}\tau^2.
    \end{align*}

    Finally, the order of the asymptotic variance is immediate from the fact that $C_{p,\epsilon_w}^2  = O((\epsilon_w^2)^{-1})$, which proves Theorem \ref{thm:joint_thorem} and  Corollary \ref{cor:conv_rate_IPW_joint}
    
\end{proof}

We now turn to estimating the asymptotic variance of $C_{p,\epsilon_w} \Tilde{\tau}_{naive}$ in \eqref{eq:CLT_IPW}. We consider the following estimators for $E_w$ and $V_w$: $\hat{E}_{w} = \frac{1}{\Tilde{N}_w} \sum_{i: \Tilde{W}_i=w}\Tilde{Y}_i$ and $\hat{V}_{w} =  \frac{1}{\Tilde{N}_w-1}\sum_{i: \Tilde{W}_i=w}(\Tilde{Y}_i - \hat{E}_{w})^2$,
where $\Tilde{N}_w = \sum_{i=1}^{N}\mathbbm{1}(\Tilde{W}_i=w)$ for $w=0,1$.
\begin{lemma}
\label{lemma:estimator_var_ipw}
$\hat{V}_{w}$ and $\hat{E}_{w}$ are consistent for  $V_w$ and $E_w$ respectively. Also, we have
\begin{align*}
    \E[\hat{E}_{w} \mid \Tilde{W}_i=w] = E_w \text{ and }
    \E[\hat{V}_{w} \mid \Tilde{W}_i=w] = V_w
\end{align*}
\end{lemma}
\begin{proof}

    \begin{align*}
      \hat{V}_1  &= \frac{1}{\Tilde{N}_1-1}\sum_{i: \Tilde{W}_i=1}(\Tilde{Y}_i - \hat{E}_1)^2 \\
      &= \frac{1}{\Tilde{N}_1-1}\sum_{i: \Tilde{W}_i=1}(\Tilde{Y}_i - \E[\Tilde{Y}_i \mid \Tilde{W}_i=1] + \E[\Tilde{Y}_i \mid \Tilde{W}_i=1] - \hat{E}_1)^2 \\ 
      &= \frac{1}{\Tilde{N}_1-1}\sum_{i: \Tilde{W}_i=1}\bigg\{(\Tilde{Y}_i - \E[\Tilde{Y}_i \mid \Tilde{W}_i=1])^2 + (\E[\Tilde{Y}_i \mid \Tilde{W}_i=1] - \hat{E}_1)^2 \\
      & \ \ \ - 2 (\Tilde{Y}_i - \E[\Tilde{Y}_i \mid \Tilde{W}_i=1])(\E[\Tilde{Y}_i \mid \Tilde{W}_i=1] - \hat{E}_1)\bigg\}\\  
      &= \frac{\Tilde{N}_1}{\Tilde{N}_1-1}\frac{1}{\Tilde{N}_1} \sum_{i: \Tilde{W}_i=1}(\Tilde{Y}_i - \E[\Tilde{Y}_i \mid \Tilde{W}_i=1])^2 - \frac{\Tilde{N}_1}{\Tilde{N}_1-1}( \hat{E}_1 - \E[\Tilde{Y}_i \mid \Tilde{W}_i=1])^2.
    \end{align*}
    Therefore,
    \begin{align*}
        \E[\hat{V}_1 \mid \Tilde{W}_i=1] &= \frac{\Tilde{N}_1}{\Tilde{N}_1-1} \Var[\Tilde{Y}_i \mid \Tilde{W}_i=1] - \frac{\Tilde{N}_1}{\Tilde{N}_1-1} \Var[\hat{E}_1 \mid \Tilde{W}_i=1] \\
        &= \frac{\Tilde{N}_1}{\Tilde{N}_1-1} \Var[\Tilde{Y}_i \mid \Tilde{W}_i=1] - \frac{1}{\Tilde{N}_1-1} \Var[\Tilde{Y}_i \mid \Tilde{W}_i=1] \\
        &= \Var[\Tilde{Y}_i \mid \Tilde{W}_i=1] \\
        &= V_1.
    \end{align*}
    
    We can follow the same procedure for $\E[\hat{V}_0  \mid \Tilde{W}_i=0]= V_0$.
\end{proof}

Using $\hat{E}_{w}$ and $\hat{V}_{w}$, we can construct the plug-in estimator for the asymptotic variance and the nominal central confidence interval at the significance level $\alpha$ as:
\begin{align*}
    \bigg(  C_{p,\epsilon_w} \Tilde{\tau}_{naive}-z_{\frac{\alpha}{2}}\sqrt{\frac{\hat{\Sigma}_{naive}}{N}} ,  C_{p,\epsilon_w} \Tilde{\tau}_{naive}+z_{\frac{\alpha}{2}}\sqrt{\frac{\hat{\Sigma}_{naive}}{N}}\bigg).
\end{align*}
where $\hat{\Sigma}_{naive} = C_{p,\epsilon_w}^2(\frac{1}{\rho_1}\hat{V}_1+\frac{1}{\rho_0}\hat{V}_0+\frac{\rho_0}{\rho_1}\hat{E}_1^2+\frac{\rho_1}{\rho_0}\hat{E}_0^2 + 2\hat{E}_0\hat{E}_1)$, which is a consistent estimator for the asymptotic variance in \eqref{eq:CLT_IPW}.

Finally, we discuss the optimality of the na\"ive estimator.

\begin{corollary}[Convergence rate]
\label{cor:conv_rate_IPW_joint}
    The na\"ive estimator under the joint scenario has the MSE $O((N\epsilon_y^2\epsilon_w^2)^{-1})$.
\end{corollary}
Setting $\epsilon_y=\epsilon_2=\epsilon/2$ gives $O((N\epsilon^4)^{-1})$. 
While we do not match the minimax lower bound of mean estimation in terms of $\epsilon$ when both $W$ and $Y$ are privatized, it should be emphasized that the estimation of PATE is significantly harder than the usual mean estimation when we do not know who belongs to which treatment group, especially using a non-interactive LDP mechanism as in the joint scenario.

\subsection{Details of Theorem \ref{thm:details_of_custom_scenario}}
\label{sec:details_of_custom_scenario}
By the standard central limit theorem, we have
\begin{equation}
\label{eq:clt_custom}
\begin{split}
    \sqrt{N}&(\Tilde{\tau}_{IPW} - \tau) \convd \mathrm{N}\left(0,\frac{\mu_1^2+\sigma^2_1}{p}+\frac{\mu_0^2+\sigma^2_0}{1-p}-\tau^2-\mu_0\mu_1 + \frac{2\Delta_A}{\epsilon_a^2} \right),  
\end{split}
\end{equation}
where $\mu_w = \E[Y_i(w)]$ and $\sigma_w^2 = \Var[Y_i(w)]$ for $w=0,1$.
We can then construct the plug-in estimator for the asymptotic variance and the nominal central confidence interval at the significance level $\alpha$ as:
\begin{align*}
    \bigg( \Tilde{\tau}_{IPW}-z_{\frac{\alpha}{2}}\sqrt{\frac{\hat{\Sigma}_{IPW}}{N}} ,  \Tilde{\tau}_{IPW}+z_{\frac{\alpha}{2}}\sqrt{\frac{\hat{\Sigma}_{IPW}}{N}}\bigg).
\end{align*}
where $\hat{\Sigma}_{IPW} = \frac{1}{N-1}\sum_{i=1}^{N}(\Tilde{A}_i - \hat{E}_A)^2$ with $\hat{E}_A=\frac{1}{N}\sum_{i=1}^{N}\Tilde{A}_i$, which is an unbiased estimator for the asymptotic variance in \eqref{eq:clt_custom}.

\subsection{Details of Theorem \ref{thm:details_of_custom2_scenario}}
\label{sec:details_of_custom2_scenario}

First, we provide the following asymptotic results regarding this estimator. 
\begin{theorem}
    \label{thm:dm_custom_asymptotic}
    $\Tilde{\tau}_{DM}$ is consistent for $\tau$ and $\sqrt{N} (\Tilde{\tau}_{DM} - \tau )$ converges in distribution to
    \begin{equation}
    \label{eq:dm_custom_asymptotic}
        \mathrm{N} \left( 0, 4\mu_0\mu_1 + \frac{\sigma_{0}^2}{1-p} + \frac{\sigma_{1}^2}{p} +\frac{2}{\epsilon_{b_1}^2} \left( \frac{\mu_0}{1-p} + \frac{\mu_1}{p} \right)^2 + \frac{2}{p^2\epsilon_{b_2}^2} + \frac{2}{(1-p)^2\epsilon_{b_3}^2} \right).
    \end{equation}
\end{theorem}
\begin{proof}
    First, we have
    \begin{align*}
        &\E[\Tilde{B}_{i,1}] = p\mu_1, \E[\Tilde{B}_{i,2}] = (1-p)\mu_0, \E[\Tilde{B}_{i,3}] = p, \E[\Tilde{B}_{i,4}] = 1-p, \Var[\Tilde{B}_{i,1}] = p\sigma_1^2+p(1-p)\mu_1^2+\frac{2}{\epsilon_{b1}^2}, \\
        &\Var[\Tilde{B}_{i,2}] = (1-p)\sigma_0^2+p(1-p)\mu_0^2+\frac{2}{\epsilon_{b2}^2}, \Var[\Tilde{B}_{i,3}] = p(1-p)+\frac{2}{\epsilon_{b3}^2}, \Var[\Tilde{B}_{i,4}] = p(1-p)+\frac{2}{\epsilon_{b3}^2}, \\
        &\Cov[\Tilde{B}_{i,1}, \Tilde{B}_{i,2}] = -p(1-p)\mu_0\mu_1, \Cov[\Tilde{B}_{i,1}, \Tilde{B}_{i,3}] = p(1-p)\mu_1, \Cov[\Tilde{B}_{i,1}, \Tilde{B}_{i,4}] = 0, \\
        &\Cov[\Tilde{B}_{i,2}, \Tilde{B}_{i,3}] = 0,  \Cov[\Tilde{B}_{i,2}, \Tilde{B}_{i,4}] = p(1-p)\mu_0,  \Cov[\Tilde{B}_{i,3}, \Tilde{B}_{i,4}] = -p(1-p)\mu_0\mu_1.  
    \end{align*}
    By the central limit theorem, we have
    
    \begin{align*}
        \sqrt{N}\begin{pmatrix}
            \frac{1}{N}\sum_{i=1}^{N}\Tilde{B}_{i,1} - \E[\Tilde{B}_{i,1}]\\
            \frac{1}{N}\sum_{i=1}^{N}\Tilde{B}_{i,2} - \E[\Tilde{B}_{i,2}]\\
            \frac{1}{N}\sum_{i=1}^{N}\Tilde{B}_{i,3} - \E[\Tilde{B}_{i,3}]\\
            \frac{1}{N}\sum_{i=1}^{N}\Tilde{B}_{i,4} - \E[\Tilde{B}_{i,4}]
        \end{pmatrix} \convd \mathrm{N} 
        \begin{pmatrix}
            \begin{pmatrix}
                0\\
                0\\
                0\\
                0
            \end{pmatrix},S^{*}
        \end{pmatrix}, 
    \end{align*}
    where 
    \begin{align*}
    S^{*} = 
        \begin{pmatrix}
                \Var[\Tilde{B}_{i,1}] & \Cov[\Tilde{B}_{i,1}, \Tilde{B}_{i,2}]  & \Cov[\Tilde{B}_{i,1}, \Tilde{B}_{i,3}]  & \Cov[\Tilde{B}_{i,1}, \Tilde{B}_{i,4}]  \\
                \Cov[\Tilde{B}_{i,2}, \Tilde{B}_{i,1}] & \Var[\Tilde{B}_{i,2}] & \Cov[\Tilde{B}_{i,2}, \Tilde{B}_{i,3}] & \Cov[\Tilde{B}_{i,2}, \Tilde{B}_{i,4}] \\
                \Cov[\Tilde{B}_{i,3}, \Tilde{B}_{i,1}] & \Cov[\Tilde{B}_{i,3}, \Tilde{B}_{i,2}] & \Var[\Tilde{B}_{i,3}] & \Cov[\Tilde{B}_{i,3}, \Tilde{B}_{i,4}] \\
                \Cov[\Tilde{B}_{i,4}, \Tilde{B}_{i,1}] & \Cov[\Tilde{B}_{i,4}, \Tilde{B}_{i,2}] & \Cov[\Tilde{B}_{i,4}, \Tilde{B}_{i,3}] & \Var[\Tilde{B}_{i,4}]
            \end{pmatrix}.
    \end{align*}
    Define a function $h(a,b,c,d)= \frac{a}{c}-\frac{b}{d}$ and $\nabla h = (\frac{\partial h}{\partial a},\frac{\partial h}{\partial b},\frac{\partial h}{\partial c},\frac{\partial h}{\partial d})$, where
    \begin{align*}
        \frac{\partial h}{\partial a} = \frac{1}{c}, \frac{\partial h}{\partial b} = -\frac{1}{d}, \frac{\partial h}{\partial c} = -\frac{a}{c^2}, \frac{\partial h}{\partial d} = \frac{b}{d^2}.
    \end{align*}
    Note that 
    \begin{align*}
        \tau = \mu_1-\mu_0= \frac{\E[\Tilde{B}_{i,1}]}{\E[\Tilde{B}_{i,3}]} - \frac{\E[\Tilde{B}_{i,2}]}{\E[\Tilde{B}_{i,4}]} = h\left(\E[\Tilde{B}_{i,1}],\E[\Tilde{B}_{i,2}],\E[\Tilde{B}_{i,3}],\E[\Tilde{B}_{i,4}] \right),
    \end{align*}
    and 
    \begin{align*}
         \Tilde{\tau}_{DM} = \frac{\sum_{i=1}^{N}\Tilde{B}_{i,1}}{\sum_{i=1}^{N}\Tilde{B}_{i,3}} - \frac{\sum_{i=1}^{N}\Tilde{B}_{i,2}}{\sum_{i=1}^{N}\Tilde{B}_{i,4}} = h\left(\frac{1}{N}\sum_{i=1}^{N}\Tilde{B}_{i,1},\frac{1}{N}\sum_{i=1}^{N}\Tilde{B}_{i,2},\frac{1}{N}\sum_{i=1}^{N}\Tilde{B}_{i,3},\frac{1}{N}\sum_{i=1}^{N}\Tilde{B}_{i,4} \right).
    \end{align*}
    By applying the delta method, we have
    \begin{align*}
        \sqrt{N}(\Tilde{\tau}_{DM}- \tau) \convd \mathrm{N} (0, \Sigma^{*}),
    \end{align*}
    where $\Sigma^{*}= \nabla h(\mathbf{E}) 'S^{*} \nabla h(\mathbf{E}) $. $\nabla h(\mathbf{E})$ denotes $\nabla h$ evaluated at $\mathbf{E}=(\E[\Tilde{B}_{i,1}],\E[\Tilde{B}_{i,2},\E[\Tilde{B}_{i,3}],\E[\Tilde{B}_{i,4}])$. Calculating $\Sigma^{*}$  proves our claim in Thereom \ref{thm:dm_custom_asymptotic}. The estimator of $\Sigma^{*}$ that we adopt in Section \ref{sec:custom2} are a plug-in estimator with consistent estimators of $\nabla h(\mathbf{E})$ and $S^{*}$.

\end{proof}

By Theorem \ref{thm:dm_custom_asymptotic}, the asymptotic variance of $\Tilde{\tau}_{DM}$ has the convergence rate $O((N(\epsilon_{b_1}^2+\epsilon_{b_2}^2+\epsilon_{b_3}^2))^{-1})$. Setting $\epsilon_{b_1}=\epsilon_{b_2}=\epsilon_{b_3} = \epsilon/3$ gives $O((N\epsilon^2)^{-1})$, which also matches the minimax lower bound for the locally private mean estimation, indicating the optimality of the estimator.
 
Let $\hat{E}_{B_j}=\frac{1}{N}\sum_{i=1}^{N}\Tilde{B}_{i,j}$, $\hat{V}_{B_j}=\frac{1}{N-1}\sum_{i=1}^{N}(\Tilde{B}_{i,j}-\hat{E}_{B_j})^2$ for $j=1,2,3,4$ and $\widehat{\mathrm{Cov}_{j,k}}=\frac{1}{N-1}\sum_{i=1}^{N}(\Tilde{B}_{i,j}-\hat{E}_{B_j})(\Tilde{B}_{i,k}-\hat{E}_{B_k})$ for $j \neq k$. Then, we construct the plug-in estimator for the asymptotic variance and the nominal central confidence interval at the significance level $\alpha$ as:
\begin{align*}
    \bigg( \Tilde{\tau}_{DM}-z_{\frac{\alpha}{2}}\sqrt{\frac{\hat{\Sigma}_{DM}}{N}} ,  \Tilde{\tau}_{DM}+z_{\frac{\alpha}{2}}\sqrt{\frac{\hat{\Sigma}_{DM}}{N}}\bigg).
\end{align*}
where $\hat{\Sigma}_{DM} = \hat{\mathbf{e}}' \hat{\mathbf{S}} \hat{\mathbf{e}}$, with $\hat{\mathbf{e}} = ( 1/\hat{E}_{B_3}, -1/(1-\hat{E}_{B_3}), -\hat{E}_{B_1}/\hat{E}_{B_3}^2, \hat{E}_{B_2}/(1-\hat{E}_{B_3})^2)'$ and 
\begin{align*}
    \hat{\mathbf{S}} = \begin{pmatrix}
            \hat{V}_{B_1} & \widehat{\mathrm{Cov}_{1,2}} & \widehat{\mathrm{Cov}_{1,3}} & \widehat{\mathrm{Cov}_{1,4}}\\
            \widehat{\mathrm{Cov}_{2,1}} & \hat{V}_{B_2} & \widehat{\mathrm{Cov}_{2,3}} & \widehat{\mathrm{Cov}_{2,4}}\\
            \widehat{\mathrm{Cov}_{3,1}} & \widehat{\mathrm{Cov}_{3,2}} & \hat{V}_{B_3}  & \widehat{\mathrm{Cov}_{3,4}}\\
            \widehat{\mathrm{Cov}_{4,1}} & \widehat{\mathrm{Cov}_{4,2}} &  \widehat{\mathrm{Cov}_{4,3}} & \hat{V}_{B_4} \\
        \end{pmatrix}.
\end{align*}
This is a consistent estimator for the asymptotic variance in \eqref{eq:dm_custom_asymptotic}.

\section{Bayesian Methodology}
\label{sec:bayes_detail}
\subsection{Details of the DPM}
\label{sec:detail_dpm}
We say the probability measure $H$ is generated from a Dirichlet Process, DP$(\alpha, H_0)$, with a concentration parameter $\alpha>0$ and  a base probability measure $H_0$ over a measurable space $(\Theta, \mathcal{S})$ \citep{Ferguson1974} if, for any finite partition $(S_1,...,S_k)$ of $\mathcal{S}$, we have
\begin{equation*}
    \big(H(S_1), ..., H(S_k)\big) \sim \text{Dir}\big(\alpha H_0(S_1), ..., \alpha H_0(S_k)\big),
\end{equation*}
where Dir$(\alpha_1,...,\alpha_k)$ denotes the Dirichlet distribution with positive parameters $\alpha_1,...,\alpha_k$.  The DPM is specified as
\begin{align*}
        \{Y_1(0),Y_1(1)\},...,\{Y_N(0),Y_N(1)\} \mid \Phi_1,...,\Phi_N &\indsim p(Y_i(0),Y_i(1)|\Phi_i),\\
        \Phi_1,...,\Phi_N | H &\indsim H, \\
        H &\indsim DP(\alpha, H_0).
\end{align*}
We write $ \displaystyle{\indsim}$ to say \emph{independently distributed}. This model has unit-level parameters $\Phi_i$ for $i=1,...,N$, but the discreteness of the Dirichlet process (DP) distributed prior implies that the vector $\mathbf{\Phi}=(\Phi_1,...,\Phi_N)$ can be rewritten in terms of its unique values $\mathbf{\Phi^*}=(\Phi_1^*,...,\Phi_K^*)$. 
In particular, this can be represented in the following stick-breaking process.
\begin{equation*}
    \begin{split}
        H = \sum_{k=1}^{\infty} u_k \delta_{\Phi_k}, \:\: u_k = v_k \prod_{l<k}[1-v_l], \:\: v_l \indsim \text{Beta}(1, \alpha).
    \end{split}
\end{equation*}
More specifically, the outcome model is specified by the following model.
\begin{equation}
\label{eq:dpm}
    \pr(Y_i(w) | \boldsymbol{\mu}, \boldsymbol{\Sigma}) \propto \sum_{k=1}^{\infty} u_k \text{TN}(\mu_{w}^k,\Sigma_{w}^k, 0, 1),
\end{equation}
where $\mathrm{TN}(\mu,\sigma^2,u,l)$ denotes the truncated normal distribution with the mean, variance, upper bound and lower bound parameters. The atoms $\Phi_k=(\mu_{0}^{k}, \mu_{1}^{k}, \Sigma_{0}^{k}, \Sigma_{1}^{k})$ and the weight parameters $u_k$ are nonparametrically specified via DP$(\alpha, H_0)$. This can be regarded as the infinite mixture of normal distributions, where $\mu_{w}^{k}$ and  $\Sigma_{w}^{k}$ is the location parameter and variance parameter of each component respectively. 

For inference, we adopt an approximated blocked Gibbs sampler based on a truncation of the stick-breaking representation of the DP proposed by \citet{Ishwaran2000}, due to its simplicity. In this algorithm, we first set a conservatively large upper bound, $K \leq \infty$, on the number of components that units potentially belong to. Let $C_i\in \{1,...,K\}$ denote the latent class indicators with a multinomial distribution, $C_i \sim MN(\mathbf{w})$ where $\mathbf{u}=(u_1,...,u_K)$ denote the weights of all components of the DPM. Conditional on $C_i=k$, \eqref{eq:dpm} is greatly simplified to
\begin{equation*}
    \pr(Y_i(w) | \boldsymbol{\mu}, \boldsymbol{\Sigma}) \propto  \text{TN}(\mu_{w}^k,\Sigma_{w}^k,0,1).
\end{equation*}
\citet{Ishwaran2001} showed that an accurate approximation to the exact DP is obtained as long as $K$ is chosen sufficiently large. The DPM provides an automatic selection mechanism for the number of active components $K^* < K$. To ensure that $K$ is sufficiently large, we run several MCMC iterations with different values of $K$. If the current iteration occupies all components, then $K$ is not large enough, so we increase $K$ for the next iteration. We conduct this iterative process until the number of the occupied components is below $K$.

\subsection{Detailed Steps of Gibbs Sampler}\label{sec:GibbsDetails}
In this section we present the detailed steps of the Gibbs sampler that is described in Section \ref{sec:gibbs_outlines}. 
The algorithm is inspired by \citet{Schwartz2011} and \citet{ohnishi_DOI_2022}.

\begin{enumerate}
    \item Given $Y_{i}(0),Y_{i}(1)$, draw each $W_i$ from 
    \begin{equation*}
        \pr(W_i=1|-) =\frac{r_1}{r_0+r_1},
    \end{equation*}
    where, for unit $i$ with $\Tilde{W}_i=0$,
    \begin{equation*}
        r_0=\mathrm{Lap(\Tilde{Y}_{i} \mid Y_{i}(0), 1/\epsilon_y)}\qeps (1-p) \text{ and } r_1=\mathrm{Lap(\Tilde{Y}_{i} \mid Y_{i}(1), 1/\epsilon_y)}(1-\qeps) p,
    \end{equation*}
    and for unit $i$ with $\Tilde{W}_i=1$,
    \begin{equation*}
        r_0=\mathrm{Lap(\Tilde{Y}_{i} \mid Y_{i}(0), 1/\epsilon_y)}(1-\qeps) (1-p) \text{ and } r_1=\mathrm{Lap(\Tilde{Y}_{i} \mid Y_{i}(1), 1/\epsilon_y)}\qeps p.
    \end{equation*}
    where $\mathrm{Lap(y \mid \mu, \sigma)}$ is the pdf of the laplace distribution evaluated at $y$  with the location parameter $\mu$ and scale parameter $\sigma$.
    \item Given $\boldsymbol{\mu}$, $\boldsymbol{\Sigma}$, $\mathbf{u}$, $C_i$ and $W_i=w$, draw $Y_i(1-w)$ according to:
    \begin{equation*}
        Y_i(1-w) \sim \mathrm{TN}(\mu_{1-w}^{C_i}, \Sigma_{1-w}^{C_i},0,1),
    \end{equation*}
    where $\mathrm{TN}(\mu,\sigma^2,u,l)$ denotes the truncated normal distribution with the mean, variance, upper bound and lower bound parameters. 
    
    Then, draw $Y_i(w)$ using the following Privacy-Aware Metropolis-within-Gibbs sampler \citet{Nianqiao_Jordan_2022}:
    \begin{enumerate}
        \item Draw a proposal: $y* \sim \mathrm{TN}(\mu_{w}^{C_i}, \Sigma_{w}^{C_i},0,1)$.
        \item Accept the proposal with probability $\alpha = \min\big(1, \frac{\mathrm{Lap(y* \mid \Tilde{Y}_i, 1 /  \epsilon_y)}}{\mathrm{Lap(y^{prev} \mid \Tilde{Y}_i, 1 / \epsilon_y)}}\big)$,
    \end{enumerate}
    where $y^{prev}$ is the value of $Y_i(w)$ in the previous step.
    \item Given $\boldsymbol{\mu}$, $\boldsymbol{\Sigma}$, $\mathbf{u}$, $Y_i(0)$ and $Y_i(1)$, draw each $C_i$ from
    \begin{align*}
        \pr(C_i=k|-) \propto u_k \mathrm{TN}(Y_i(0) \mid \mu_{0}^{k}, \Sigma_{0}^{k},0,1) \mathrm{TN}(Y_i(1) \mid \mu_{1}^{k}, \Sigma_{1}^{k},0,1).
    \end{align*}
    This is a multinomial distribution.
    \item Let $u_K'=1$. Given $\alpha$, $\mathbf{C}$, draw $u_k'$ for $k \in \{1,...,K-1\}$ from
    \begin{equation*}
        \pr(u_k'|-) \propto \text{Beta} \bigg(1+ \sum_{i:C_i=k}1, \alpha+ \sum_{i:C_i>k}1 \bigg).
    \end{equation*}
    Then, update $u_k= u_k'\prod_{j<k}(1-u_j')$.
    \item Given $\mathbf{C}$ and $\mathbf{u}'$, draw $\alpha$ from 
    \begin{equation*}
        \pr(\alpha | -) \propto \pr(\alpha) \prod_{k=1}^{K} f \bigg(u_k' \bigg| 1+ \sum_{i:C_i=k}1, \alpha+ \sum_{i:C_i>k}1 \bigg),
    \end{equation*}
    where $f$ is the pdf of $u_k'$, the beta distribution. The Metropolis-Hastings algorithm is used for this step with a proposal distribution $\mathrm{TN}(\alpha^{prev}, 1.0, 0, \infty)$. $\alpha^{prev}$ is the value of $\alpha$ in the previous step.
    \item Given $\mathbf{Y}(0)$, $\mathbf{Y}(1)$ and $\mathbf{C}$, draw $\boldsymbol{\mu}$ and $\boldsymbol{\Sigma}$ from
    \begin{enumerate}
        \item If $N_k=\sum_{i=1}^{N}\mathbbm{1}(C_i=k)>0$, draw $\Sigma_w^{k}$ from $\mathrm{IG}(2+0.5N_k, 0.2^2+0.5s_w^{k})$ where $s_w^{k}=\sum_{i:C_i=k}(Y_i(w) - \mu_w^{k})^2$ for $w=0,1$. If $N_k=0$,then draw $\Sigma_w^{k}$ from the prior $\mathrm{IG}(2,0.2^2)$.
        \item If $N_k>0$, draw $\mu_w^{k}$ from 
        \begin{align*}
            \mathrm{TN}\bigg(\frac{0.5*\Sigma_w^{k}+9.0s_w}{\Sigma_w^{k}+9.0N_k}, \frac{9.0\Sigma_w^{k}}{\Sigma_w^{k}+9.0N_k}, 0, 1\bigg),
        \end{align*}
        where $s_w=\sum_{i=1}^{N}Y_{i}(w)$.  
        If $N_k=0$, draw $\mu_w^{k}$ from
        \begin{align*}
            \mathrm{TN}(0.5, 9.0, 0, 1).
        \end{align*}
        We use a common choice of the base measure $H_0$: the Normal-Inverse-Gamma conjugate $\mathrm{N}(\mu_0,\sigma_0^2)\mathrm{N}(\mu_0,\sigma_0^2)\mathrm{IG}(a_0,b_0)\mathrm{IG}(a_0,b_0)$. 
        The specific values of the hyperparameters in this step are: $\mu_0=0.5$, $\sigma_0=3.0$, $a_0=2.0$ and $b_0=0.2^2$ for both $w=0,1$.
    \end{enumerate}
\end{enumerate}

\subsection{Modifications for Custom Scenario in Section \ref{sec:customscenario}}
\label{sec:modification_custom1}
We need to modify Step 1 and 2 for the custom scenarios. Particularly, 
\begin{enumerate}
    \item Given $Y_{i}(0),Y_{i}(1)$, draw each $W_i$ from 
    \begin{equation*}
        \pr(W_i=1|-) =\frac{r_1}{r_0+r_1},
    \end{equation*}
    where $r_w=\pr(\Tilde{A}_{i} \mid Y_{i}(0),Y_{i}(1), W_{i}=w) \pr(W_i=w)$ for $w=0,1$. Specifically, since $\Tilde{A}_{i}$ is generated by privatizing either $-Y_{i}(0)/(1-p)$ or $Y_{i}(1)/p$ given the value of $W_i$, $\pr(\Tilde{A}_{i} \mid Y_{i}(0),Y_{i}(1), W_{i}=w)=\mathrm{Lap}(\Tilde{A}_{i} \mid -Y_{i}(0)/(1-p), \Delta_a/\epsilon_a)$ for $W_i=0$, and $\pr(\Tilde{A}_{i} \mid Y_{i}(0),Y_{i}(1), W_{i}=w)=\mathrm{Lap}(\Tilde{A}_{i} \mid Y_{i}(1)/p, \Delta_a/\epsilon_a)$ for $W_i=1$.
    
    \item Given $\boldsymbol{\mu}$, $\boldsymbol{\Sigma}$, $\mathbf{u}$, $C_i$ and $W_i$, draw each $Y_i(0)$ and $Y_i(1)$ according to:
    \newline  
    \begin{equation*}
    \begin{split}
        &\pr(Y_i(W_i)|-) \propto \pr(Y_i(W_i) \mid \mu_{W_i}^{C_i}, \Sigma_{W_i}^{C_i})\pr(\Tilde{A}_{i} \mid Y_{i}(W_i)) \\
        &\pr(Y_i(1-W_i)|-) \propto \pr(Y_i(1-W_i) \mid \mu_{1-W_i}^{C_i}, \Sigma_{1-W_i}^{C_i}).
    \end{split}
    \end{equation*}
    Specifically, $\pr(\Tilde{A}_{i} \mid Y_{i}(W_i)) = \mathrm{Lap}(\Tilde{A}_{i} \mid -Y_{i}(0)/(1-p), \Delta_a/\epsilon_a)$ for $W_i=0$ and $\pr(\Tilde{A}_{i} \mid Y_{i}(W_i)) = \mathrm{Lap}(\Tilde{A}_{i} \mid Y_{i}(1)/p, \Delta_a/\epsilon_a)$ for $W_i=1$. 
    The privacy-aware Metropolis-within-Gibbs algorithm \citep{Nianqiao_Jordan_2022} is used for the draw of $Y_i(W_i)$.
\end{enumerate}

\subsection{Modifications for Custom Scenario in Section \ref{sec:custom2}}
Under the custom scenario in Section \ref{sec:custom2}, we do not have access to $p$. Therefore, we need an additional step to infer $p$. Specifically, with a prior distribution $p \sim \mathrm{Beta}(1,1)$, we add the following step.
\begin{enumerate}\addtocounter{enumi}{-1}
	\item Draw $p \sim \mathrm{Beta}\left(1+\sum_{i=1}^{N}\mathbbm{1}(W_i=1),1+\sum_{i=1}^{N}\mathbbm{1}(W_i=0)\right)$.
\end{enumerate}
Then we proceed as follows.

\begin{enumerate}
    \item Given $Y_{i}(0),Y_{i}(1)$ and $p$, draw each $W_i$ from 
    \begin{equation*}
        \pr(W_i=1|-) =\frac{r_1}{r_0+r_1},
    \end{equation*}
    where
    \begin{align*}
        r_w &= p^w(1-p)^{1-w} \pr(\Tilde{B}_{i,1} \mid Y_{i}(0),Y_{i}(1), W_{i}=w) \pr(\Tilde{B}_{i,2} \mid Y_{i}(0),Y_{i}(1), W_{i}=w) \\
        &\times \pr(\Tilde{B}_{i,3} \mid Y_{i}(0),Y_{i}(1), W_{i}=w)
    \end{align*}
    for $w=0,1$. Specifically, considering the privatization of $\Tilde{B}_{i,1}$, $\Tilde{B}_{i,2}$ and $\Tilde{B}_{i,3}$, we have 
    $ \pr(\Tilde{B}_{i,1} \mid Y_{i}(0),Y_{i}(1), W_{i}=0)=\mathrm{Lap}(\Tilde{B}_{i,1} \mid 0, 1/\epsilon_{b_2})$, 
    $ \pr(\Tilde{B}_{i,2} \mid Y_{i}(0),Y_{i}(1), W_{i}=0)=\mathrm{Lap}(\Tilde{B}_{i,2} \mid Y_{i}(0), 1/\epsilon_{b_2})$, 
    $ \pr(\Tilde{B}_{i,3} \mid Y_{i}(0),Y_{i}(1), W_{i}=0)=\mathrm{Lap}(\Tilde{B}_{i,3} \mid 0, 1/\epsilon_{b_3})$,  
    $ \pr(\Tilde{B}_{i,1} \mid Y_{i}(0),Y_{i}(1), W_{i}=1)=\mathrm{Lap}(\Tilde{B}_{i,1} \mid Y_{i}(1), 1/\epsilon_{b_1})$, 
    $ \pr(\Tilde{B}_{i,2} \mid Y_{i}(0),Y_{i}(1), W_{i}=1)=\mathrm{Lap}(\Tilde{B}_{i,2} \mid 0, 1/\epsilon_{b_2})$ and 
    $ \pr(\Tilde{B}_{i,3} \mid Y_{i}(0),Y_{i}(1), W_{i}=1)=\mathrm{Lap}(\Tilde{B}_{i,3} \mid 0, 1/\epsilon_{b_3})$.
    \item Given $\boldsymbol{\mu}$, $\boldsymbol{\Sigma}$, $\mathbf{u}$, $C_i$ and $W_i$, draw each $Y_i(0)$ and $Y_i(1)$ according to:
    \newline  
    \begin{equation*}
    \begin{split}
        &\pr(Y_i(W_i)|-) \propto \pr(Y_i(W_i) \mid \mu_{W_i}^{C_i}, \Sigma_{W_i}^{C_i})\pr(\Tilde{\mathbf{B}}_{i} \mid Y_{i}(W_i)) \\
        &\pr(Y_i(1-W_i)|-) \propto \pr(Y_i(1-W_i) \mid \mu_{1-W_i}^{C_i}, \Sigma_{1-W_i}^{C_i}).
    \end{split}
    \end{equation*}
    Specifically, $\pr(\Tilde{\mathbf{B}}_{i} \mid Y_{i}(W_i)) = \pr(\Tilde{B}_{i,2} \mid Y_{i}(0))=\mathrm{Lap}(\Tilde{B}_{i,2} \mid Y_{i}(0), 1/\epsilon_{b_2})$ for $W_i=0$ and $\pr(\Tilde{\mathbf{B}}_{i} \mid Y_{i}(W_i)) = \pr(\Tilde{B}_{i,1} \mid Y_{i}(1))=\mathrm{Lap}(\Tilde{B}_{i,1} \mid Y_{i}(1), 1/\epsilon_{b_1})$ for $W_i=1$.
    The privacy-aware Metropolis-within-Gibbs algorithm \citep{Nianqiao_Jordan_2022} is used for the draw of $Y_i(W_i)$.
\end{enumerate}

\section{Simulation Details}

\subsection{Beta GLM}
\begin{figure*}
    \centering
    \includegraphics[scale=0.6]{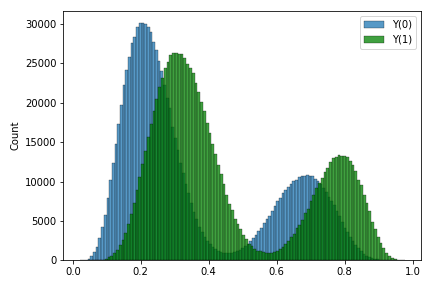}
    \caption{Distributions of $Y(0)$ and $Y(1)$ for simulation studies.}
    \label{fig:distribution}
\end{figure*}
Under the data-generating processes and the re-parameterizations of the Beta regression provided in Section \ref{sec:simulationsetting}, we generated $1000000$ samples for $Y(0)$ and $Y(1)$ to see what the data looks like. Figure \ref{fig:distribution} shows the distributions of each potential outcome. Also, the expectations of each potential outcome are expressed as:
\begin{align*}
    \E[Y(0)] &= \E_{X_1,X_2,X_3}[\mu(0)] \\
    &= \E_{X_1,X_2,X_3}\bigg[\frac{\exp(1.0-0.8X_1+0.5X_2-2.0X_3)}{1+\exp(1.0-0.8X_1+0.5X_2-2.0X_3)}\bigg]\\
    &= 0.359613,\\
    \E[Y(1)] &= \E_{X_1,X_2,X_3}[\mu(1)] \\
    &= \E_{X_1,X_2,X_3}\bigg[\frac{\exp(1.5-0.8X_1+0.5X_2-2.0X_3)}{1+\exp(1.5-0.8X_1+0.5X_2-2.0X_3)}\bigg]\\
    &= 0.457068.
\end{align*}
We refer readers to \citet{Cribari-Neto_2004} for further details about the Beta regression.

\subsection{Additional Simulations}
Table \ref{tab:simulation_freq_N100_IPW} -- \ref{tab:simulation_bayes_N1000_IPW} display the simulation results for smaller sample sizes of $N=100$ and $N=1000$. All scenarios achieve roughly $95\%$ coverage. Regarding Bias and MSE, custom scenarios demonstrate superior performance compared to the joint scenario, consistent with the observations in the main manuscript for $N=10000$. As expected, the MSE of the frequentists estimators for $N=1000$ is about $10$ times that of $N=10000$, and $N=100$ is about $10$ times that of $N=1000$, which confirms the validity of the convergence rates we derived. All discussions regarding the comparison between the frequentist and Bayesian estimators in the main manuscript are applicable to the case of $N=100,1000$. Please refer to Section \ref{sec:simulationstudies} in the main manuscript for a detailed discussion on this matter.

\begin{table*}
	\centering
	\caption{Evaluation metrics of frequentist estimators for $N=100, N_{sim}=2000$.}
	\begin{adjustbox}{width=16.cm}
		\begin{tabular}{rrrrrrrrrrrrr}
			\toprule
			  &\multicolumn{3}{c}{Coverage}   & \multicolumn{3}{c}{Bias}   & \multicolumn{3}{c}{MSE}   & \multicolumn{3}{c}{Interval Width} \\
			\cmidrule(lr){2-4} \cmidrule(lr){5-7} \cmidrule(lr){8-10} \cmidrule(lr){11-13} 
			 $\epsilon_{\mathrm{tot}}$ &  Joint & Custom (IPW) & Custom (DM) & Joint & Custom (IPW) & Custom (DM) & Joint & Custom (IPW) & Custom (DM) & Joint & Custom (IPW) & Custom (DM) \\ \hline
			0.1 & $94.75\%$ & $94.5\%$ & $100.0\%$ & $0.9025$ & $-1.0975$ & $-1.0975$ & $0.9977$ & $0.8231$ & $0.7357$ & $1.895$ & $1.881$ & $2.0$ \\
0.3 & $93.4\%$ & $94.95\%$ & $100.0\%$ & $0.9025$ & $-0.6965$ & $0.9025$ & $0.9827$ & $0.4956$ & $0.7877$ & $1.869$ & $1.803$ & $2.0$ \\
1.0 & $94.8\%$ & $95.45\%$ & $99.8\%$ & $-0.0535$ & $-0.2887$ & $0.9025$ & $0.7787$ & $0.084$ & $0.7476$ & $1.883$ & $1.137$ & $1.986$ \\
3.0 & $94.65\%$ & $94.70\%$ & $97.6\%$ & $-0.3555$ & $0.1052$ & $-0.844$ & $0.1037$ & $0.0176$ & $0.2508$ & $1.237$ & $0.518$ & $1.673$ \\
10 & $95.85\%$ & $95.0\%$ & $95.7\%$ & $0.1429$ & $0.1127$ & $-0.1147$ & $0.0115$ & $0.0092$ & $0.0226$ & $0.433$ & $0.38$ & $0.591$ \\
			\bottomrule
		\end{tabular}
	\end{adjustbox}
	\label{tab:simulation_freq_N100_IPW}
\end{table*}

\begin{table*}
	\centering
	\caption{Evaluation metrics of Bayesian estimators for $N=100, N_{sim}=1000$.}
	\begin{adjustbox}{width=16.cm}
		\begin{tabular}{rrrrrrrrrrrrr}
			\toprule
			  &\multicolumn{3}{c}{Coverage}   & \multicolumn{3}{c}{Bias}   & \multicolumn{3}{c}{MSE}   & \multicolumn{3}{c}{Interval Width} \\
			\cmidrule(lr){2-4} \cmidrule(lr){5-7} \cmidrule(lr){8-10} \cmidrule(lr){11-13} 
			 $\epsilon_{\mathrm{tot}}$ &  Joint & Custom (IPW) & Custom (DM) & Joint & Custom (IPW) & Custom (DM) & Joint & Custom (IPW) & Custom (DM) & Joint & Custom (IPW) & Custom (DM) \\ \hline
			0.1 & $100.0\%$ & $100.0\%$ & $100.0\%$ & $-0.0961$ & $-0.0958$ & $-0.0966$ & $0.00927$ & $0.00938$ & $0.00937$ & $0.616$ & $0.614$ & $0.575$ \\
0.3 & $100.0\%$ & $100.0\%$ & $100.0\%$ & $-0.0958$ & $-0.0886$ & $-0.097$ & $0.00922$ & $0.00889$ & $0.00948$ & $0.616$ & $0.602$ & $0.586$ \\
1.0 & $100.0\%$ & $100.0\%$ & $100.0\%$ & $-0.0952$ & $-0.0691$ & $-0.0951$ & $0.00932$ & $0.00939$ & $0.00961$ & $0.615$ & $0.528$ & $0.58$ \\
3.0 & $99.5\%$ & $97.6\%$ & $99.5\%$ & $-0.0657$ & $-0.0403$ & $-0.0864$ & $0.01055$ & $0.00734$ & $0.01126$ & $0.521$ & $0.367$ & $0.54$ \\
10 & $94.4\%$ & $96.7\%$ & $94.3\%$ & $-0.0155$ & $-0.0256$ & $-0.0259$ & $0.00343$ & $0.00241$ & $0.00676$ & $0.232$ & $0.198$ & $0.304$ \\
			\bottomrule
		\end{tabular}
	\end{adjustbox}
	\label{tab:bayes_N100_IPW}
\end{table*}


\begin{table*}
	\centering
	\caption{Evaluation metrics of frequentist estimators for $N=1000, N_{sim}=2000$.}
	\begin{adjustbox}{width=16.cm}
		\begin{tabular}{rrrrrrrrrrrrr}
			\toprule
			  &\multicolumn{3}{c}{Coverage}   & \multicolumn{3}{c}{Bias}   & \multicolumn{3}{c}{MSE}   & \multicolumn{3}{c}{Interval Width} \\
			\cmidrule(lr){2-4} \cmidrule(lr){5-7} \cmidrule(lr){8-10} \cmidrule(lr){11-13} 
			 $\epsilon_{\mathrm{tot}}$ &  Joint & Custom (IPW) & Custom (DM) & Joint & Custom (IPW) & Custom (DM) & Joint & Custom (IPW) & Custom (DM) & Joint & Custom (IPW) & Custom (DM) \\ \hline
			0.1 & $94.40\%$ & $94.55\%$ & $100.0\%$ & $-1.0975$ & $-0.2006$ & $0.1042$ & $1.0006$ & $0.4808$ & $0.7886$ & $1.887$ & $1.781$ & $1.999$ \\
0.3 & $95.05\%$ & $95.85\%$ & $99.55\%$ & $0.9025$ & $0.1541$ & $0.9025$ & $0.9411$ & $0.0857$ & $0.7609$ & $1.901$ & $1.148$ & $1.987$ \\
1.0 & $94.20\%$ & $95.20\%$ & $98.0\%$ & $-1.0975$ & $-0.1484$ & $0.0074$ & $0.3919$ & $0.0089$ & $0.2192$ & $1.737$ & $0.369$ & $1.618$ \\
3.0 & $95.15\%$ & $94.65\%$ & $96.3\%$ & $-0.0214$ & $0.0245$ & $-0.0585$ & $0.0111$ & $0.0018$ & $0.0216$ & $0.411$ & $0.164$ & $0.586$ \\
10 & $94.55\%$ & $95.05\%$ & $94.65\%$ & $-0.0082$ & $-0.0189$ & $0.0671$ & $0.0012$ & $0.0009$ & $0.0022$ & $0.137$ & $0.12$ & $0.181$ \\
			\bottomrule
		\end{tabular}
	\end{adjustbox}
	\label{tab:simulation_freq_N1000_IPW}
\end{table*}

\begin{table*}
	\centering
	\caption{Evaluation metrics of Bayesian estimators for $N=1000, N_{sim}=1000$.}
	\begin{adjustbox}{width=16.cm}
		\begin{tabular}{rrrrrrrrrrrrr}
			\toprule
			  &\multicolumn{3}{c}{Coverage}   & \multicolumn{3}{c}{Bias}   & \multicolumn{3}{c}{MSE}   & \multicolumn{3}{c}{Interval Width} \\
			\cmidrule(lr){2-4} \cmidrule(lr){5-7} \cmidrule(lr){8-10} \cmidrule(lr){11-13} 
			 $\epsilon_{\mathrm{tot}}$ &  Joint & Custom (IPW) & Custom (DM) & Joint & Custom (IPW) & Custom (DM) & Joint & Custom (IPW) & Custom (DM) & Joint & Custom (IPW) & Custom (DM) \\ \hline
			0.1 & $100.0\%$ & $100.0\%$ & $100.0\%$ & $-0.0973$ & $-0.0935$ & $-0.0987$ & $0.0096$ & $0.0094$ & $0.0099$ & $0.45$ & $0.448$ & $0.462$ \\
0.3 & $100.0\%$ & $99.1\%$ & $100.0\%$ & $-0.0972$ & $-0.0787$ & $-0.0974$ & $0.0096$ & $0.0088$ & $0.0099$ & $0.455$ & $0.411$ & $0.463$ \\
1.0 & $100.0\%$ & $94.0\%$ & $100.0\%$ & $-0.0928$ & $-0.0365$ & $-0.0916$ & $0.0093$ & $0.0046$ & $0.0099$ & $0.445$ & $0.259$ & $0.433$ \\
3.0 & $96.1\%$ & $92.3\%$ & $95.7\%$ & $-0.0265$ & $-0.0173$ & $-0.0469$ & $0.004$ & $0.0015$ & $0.0056$ & $0.258$ & $0.134$ & $0.303$ \\
10 & $92.7\%$ & $95.5\%$ & $92.0\%$ & $-0.0103$ & $-0.0074$ & $-0.0144$ & $0.0004$ & $0.0003$ & $0.0009$ & $0.078$ & $0.061$ & $0.107$ \\
			\bottomrule
		\end{tabular}
	\end{adjustbox}
	\label{tab:simulation_bayes_N1000_IPW}
\end{table*}


\section{Regression Adjustment}
\label{sec:regression_adjustment}
\subsection{Overview}
In the context of randomized experiments, causal effects $\tau$ can be identified solely using the treatment assignment and outcome variables. Also, as demonstrated in prior sections, our custom frequentist estimators achieve minimax optimality without the need for covariates.
However, there is a clear rationale for incorporating covariates when deducing causal effects in randomized settings: they can enhance the efficiency of inference by leveraging pertinent individual data. This enhancement method is termed regression adjustment \citep{Lin2013}. Nevertheless, applying regression adjustment within LDP presents challenges. Specifically, it could incur additional privacy costs for the covariates, and these costs could escalate significantly for high-dimensional covariates.
In this section, we present another type of frequentist estimator for the joint scenario, namely the OLS estimator. We explore its advantages and constraints, compared to the IPW estimator in the same scenario.

Assume that the observed covariates are privatized by the Laplace mechanism. We assume  $X_{i,j} \in [0,1]$ for $i=1,\ldots, N$ and $j=1,\ldots, d$ to ensure bounded $\ell_1$-sensitivity. The privatized outcomes and covariates are $\Tilde{X}_{i,j} = X_{i,j} + \nu_{i,j}^{X},$
where $\nu_{i,j}^{X} \iidsim \mathrm{Lap}(d/\epsilon_x)$.
By composition, the joint release of $(\Tilde Y_i,\Tilde X_{i,1},\ldots, \Tilde X_{i,d},\Tilde W_i)_{i=1}^N$ satisfies $(\epsilon_y+\epsilon_x+\epsilon_w)$-LDP. 

 Without privacy considerations, it is well known that the covariate adjustment can further improve the efficiency, even without assuming a correctly specified outcome model \citep{Lin2013}. Specifically, we propose the following plug-in OLS estimator.
\begin{equation}
\label{eq:estimator_ols}
    \Tilde{\tau}_{OLS} =  \Tilde{\alpha}_{(1)} - \Tilde{\alpha}_{(0)} + \Bar{X}(\Tilde{\beta}_{(1)} - \Tilde{\beta}_{(0)}),
\end{equation}
 where $\Bar{X} = \frac{1}{N} \sum_{i=1}^{N} \Tilde{X}_i$ and $(\Tilde{\alpha}_{(w)}, \Tilde{\beta}_{(w)}) = \argmin_{\alpha, \beta} \sum_{i: \Tilde{W}_i=w}(\Tilde{Y}_{i}-\alpha-\Tilde{X}_i'\beta)^2$ for $w=0,1$. 
Note that, under some regularity conditions \citep[p.~440]{LehmCase98}, $(\Tilde{\alpha}_{(w)}, \Tilde{\beta}_{(w)})$ converges to $(\Tilde{\alpha}_{(w)}^{*}, \Tilde{\beta}_{(w)}^{*})$, defined as
\begin{equation*}
    (\Tilde{\alpha}_{(w)}^{*}, \Tilde{\beta}_{(w)}^{*}) = \argmin_{\alpha, \beta} \E[(\Tilde{Y}_{i}-\alpha-\Tilde{X}_i'\beta)^2 \mid \Tilde{W}_i=w].
\end{equation*}

We investigate the potential bias of the na\"ive OLS estimator and propose a bias-corrected version. The following theorem states that the na\"ive OLS estimator \eqref{eq:estimator_ols} is an inconsistent estimator for $\tau$, but multiplying by the same factor $C_{p,\epsilon_w}$ makes it consistent. The central limit theorem has also been developed. 

\begin{theorem}
\label{thm:properties_of_ols}
    \begin{enumerate}
        \item (Consistency) $C_{p,\epsilon_w} \Tilde{\tau}_{OLS}$ is consistent for $\tau$.
        \item (CLT) $\sqrt{N}(C_{p,\epsilon_w} \Tilde{\tau}_{OLS} - \tau)$ converges in distribution to 
        \begin{equation}\label{eq:CLT_OLS}
            \mathrm{N} \bigg(0, C_{p,\epsilon_w}^2\bigg(\frac{\mathrm{MSE}_1}{\rho_1}+\frac{\mathrm{MSE}_0}{\rho_0} \bigg)\bigg),
        \end{equation}
        where $\mathrm{MSE}_w = \E[(\Tilde{Y}_i - \Tilde{\alpha}_{(w)}^{*} - \Tilde{X}_i'\Tilde{\beta}^{*}_{(w)})^2 \mid \Tilde{W}_i=w]$ for $w=0,1$.
        \item (Confidence Interval) The following interval is  the nominal
central confidence at the significance level $\alpha$:
        \begin{align*}
            \bigg(  C_{p,\epsilon_w} \Tilde{\tau}_{OLS}-z_{\frac{\alpha}{2}}\sqrt{\frac{\hat{\Sigma}_{OLS}}{N}} ,  C_{p,\epsilon_w} \Tilde{\tau}_{OLS}+z_{\frac{\alpha}{2}}\sqrt{\frac{\hat{\Sigma}_{OLS}}{N}}\bigg),
        \end{align*}
        where 
        $\hat{\Sigma}_{OLS}= C_{p,\epsilon_w}^2\big(\frac{\widehat{\mathrm{MSE}}_1}{\rho_1}+\frac{\widehat{\mathrm{MSE}}_0}{\rho_0} \big)$ and $\widehat{\mathrm{MSE}}_w = \frac{1}{\Tilde{N}_w}\sum_{i:\Tilde{W}_i=w}(\Tilde{Y}_i-\Tilde{\alpha}_{(w)}-\Tilde{X}_i\Tilde{\beta}_{(w)})^2$ for $w=0,1$.
    \end{enumerate}
\end{theorem}
\begin{proof}
    Consider the objective function
    \begin{align*}
        \mathcal{Q}(\alpha_{(w)}, \beta_{(w)})&=\E[(\Tilde{Y}_{i}-\alpha_{(w)}-\Tilde{X}_i^{'}\beta_{(w)})^2 \mid \Tilde{W}_i=w] \\
        &= \E[(\Tilde{Y}_{i}-\gamma_{(w)}-(\Tilde{X}_i^{'} - \mu_{\Tilde{X}})\beta_{(w)})^2 \mid \Tilde{W}_i=w],
    \end{align*}
    where $\gamma_{(w)}=\alpha_{(w)} + \mu_{\Tilde{X}} \beta_{(w)}$. Note that, for both $w=0, 1$,
    \begin{align*}
        \mu_{\Tilde{X}} = \E[\Tilde{X}_i \mid \Tilde{W}_i =w]=\E[X_i \mid \Tilde{W}_i =w]=\E[X_i]=\mu_{X}.
    \end{align*}
    The second equality follows from the independence of noise $\nu_i^{X}$, and the third equality follows from the randomized assignment of $W_i$ and the independence of the randomized response mechanism.
    Minimizing the right-hand side over $\gamma_{(w)}$ and $\beta_{(w)}$ leads to the same values for $\alpha_{(w)}$ and $\beta_{(w)}$ as minimizing the left-hand side over $\alpha_{(w)}$ and $\beta_{(w)}$, with the least squares estimate of $\gamma_{(w)}^{*}=\alpha_{(w)}^{*}+\mu_{\Tilde{X}}\beta_{(w)}^{*}$. 
    \begin{align*}
        &\mathcal{Q}(\gamma_{(w)}, \beta_{(w)})\\
        &= \E[(\Tilde{Y}_i-\gamma_{(w)}-(\Tilde{X}_i^{'} - \mu_X)\beta_{(w)})^2 \mid \Tilde{W}_i=w] \\
        &= \E[(\Tilde{Y}_i-\gamma_{(w)})^2 \mid \Tilde{W}_i=w] + \E[((\Tilde{X}_i^{'} - \mu_{\Tilde{X}})\beta_{(w)})^2 \mid \Tilde{W}_i=w] - 2 \E[ (\Tilde{Y}_i-\gamma_{(w)})(\Tilde{X}_i^{'} - \mu_{\Tilde{X}})\beta_{(w)}\mid \Tilde{W}_i=w] \\
        &= \E[(\Tilde{Y}_i-\gamma_{(w)})^2 \mid \Tilde{W}_i=w] + \E[((\Tilde{X}_i^{'} - \mu_{\Tilde{X}})\beta_{(w)})^2 \mid \Tilde{W}_i=w] - 2 \E[ \Tilde{Y}_i(\Tilde{X}_i^{'} - \mu_{\Tilde{X}})\beta_{(w)}\mid \Tilde{W}_i=w].
    \end{align*}
    The last two terms do not depend on $\gamma_{(w)}$. Thus, minimizing $\mathcal{Q}(\gamma_{(w)}, \beta_{(w)})$ over $\gamma_{(w)}$ is equivalent to minimizing $\E[(\Tilde{Y}_{i}-\gamma_{(w)})^2 \mid \Tilde{W}_i=w]$ over $\gamma_{(w)}$, which leads to the minimizer
    \begin{align*}
        \Tilde{\gamma}_{(1)}^{*} &= \E[\Tilde{Y}_i | \Tilde{W}_i=1]= \E[Y_i | \Tilde{W}_i=1] \\
        &= \sum_{w=0}^{1}\E[Y_i| \Tilde{W}_i=1, W_i=w] \pr(W_i=w \mid \Tilde{W}_i=1) \\
        &= \frac{\pbar\qbareps}{\pq}\E[Y_i(0)] + \frac{p\qeps}{\pq}\E[Y_i(1)]. \\
    \end{align*}
    Similarly, we have
    \begin{align*}
        \Tilde{\gamma}_{(0)}^{*}= \frac{\pbar\qeps}{\pqbar}\E[Y_i(0)] - \frac{p\qbareps}{\pqbar}\E[Y_i(1)].\\
    \end{align*}
    Then, we have
    \begin{align*}
        \Tilde{\gamma}_{(1)}^{*} - \Tilde{\gamma}_{(0)}^{*} =& \frac{(\qeps - \qbareps)p\pbar}{(\pqbar)(\pq)}(\E[Y_i(1)] - \E[Y_i(0)] )\\
        =& \frac{(\qeps - \qbareps)p\pbar}{(\pqbar)(\pq)}\tau \\
        =& \frac{1}{ C_{p,\epsilon_w}}\tau.
    \end{align*}
    Finally, noting the fact that $\Tilde{\gamma}_{(w)}^{*} = \Tilde{\alpha}_{(w)}^{*}+\mu_{\Tilde{X}}\Tilde{\beta}_{(w)}^{*}$
    and, under some regularity conditions, $(\Tilde{\alpha}_{(w)}, \Tilde{\beta}_{(w)})$ converges to $(\Tilde{\alpha}_{(w)}^{*}, \Tilde{\beta}_{(w)}^{*})$, 
    $$\Tilde{\tau}_{OLS} =\Tilde{\alpha}_{(1)} - \Tilde{\alpha}_{(0)} + \Bar{\Tilde{X}}(\Tilde{\beta}_{(1)} - \Tilde{\beta}_{(0)}) \convp \Tilde{\gamma}_{(1)}^{*} - \Tilde{\gamma}_{(0)}^{*}= \frac{1}{ C_{p,\epsilon_w}}\tau.$$
    Thus, by the continuous mapping theorem, $ C_{p,\epsilon_w}\Tilde{\tau}_{OLS}$ is a consistent estimator for $\tau$.

    Next, we obtain the central limit theorem. Again, it is convenient to parameterize the model using $(\gamma_w, \beta_w)$ instead of $(\alpha_w, \beta_w)$. In terms of these parameters, the objective function for $\Tilde{W}_i=w$ is 
    \begin{align*}
        \sum_{i:\Tilde{W}_i=w} \big(\Tilde{Y}_i - \gamma - (\Tilde{X}_i - \mu_{\Tilde{X}}) \beta\big)^2.
    \end{align*}
    The first order conditions for the estimators $(\Tilde{\gamma}_w, \Tilde{\beta}_w)$ are 
    \begin{align*}
        \sum_{i:\Tilde{W}_i=w} \psi(\Tilde{Y}_i, {\Tilde{X}}_i, \Tilde{\gamma}_w, \Tilde{\beta}_w)=0,
    \end{align*}
    where $\psi(\cdot)$ is a two-component column vector:
    \begin{align*}
        \psi(y, x, \gamma, \beta) = \begin{pmatrix}
            y-\gamma-(x - \mu_{\Tilde{X}}) \beta \\
            (x - \mu_{\Tilde{X}}) (y-\gamma-(x - \mu_{\Tilde{X}}) \beta)
        \end{pmatrix}.
    \end{align*}
    The standard M-estimation results imply that, under standard regularity conditions, the estimator is consistent and asymptotically normally distributed:
    \begin{align*}
        \sqrt{N_w}\begin{pmatrix}
            \Tilde{\gamma}_w - \Tilde{\gamma}_w^{*}\\
            \Tilde{\beta}_w - \Tilde{\beta}_w^{*}
        \end{pmatrix} \convd \mathrm{N} 
        \begin{pmatrix}
            \begin{pmatrix}
                0\\
                0
            \end{pmatrix},
            \Gamma_w^{-1} \Delta_w (\Gamma_w^{'})^{-1}
        \end{pmatrix},
    \end{align*}
    where $N_w = \sum_{i=1}^{N}\mathbbm{1}(\Tilde{W}_i = w)$ and the two components of the covariance matrix are
    \begin{align*}
        \Gamma_w &= \left.\E\bigg[\frac{\partial}{\partial (\gamma, \beta)} \psi(\Tilde{Y}_i, \Tilde{X}_i, \gamma, \beta) \mid \Tilde{W}_i=w \bigg]\right\vert_{(\Tilde{\gamma}_w^{*}, \Tilde{\beta}_w^{*})} \\
        &= \E\bigg[ \begin{pmatrix}
            -1 & -(\Tilde{X}_i - \mu_{\Tilde{X}}) \\
            -(\Tilde{X}_i - \mu_{\Tilde{X}})^{'} & -(\Tilde{X}_i - \mu_{\Tilde{X}})^{'} (\Tilde{X}_i - \mu_{\Tilde{X}})
        \end{pmatrix} \mid \Tilde{W}_i=w\bigg] \\
        &= \E\bigg[ \begin{pmatrix}
            -1 & 0\\
            0 & -\E[(\Tilde{X}_i - \mu_{\Tilde{X}})^{'} (\Tilde{X}_i - \mu_{\Tilde{X}})]
        \end{pmatrix} \mid \Tilde{W}_i=w\bigg],
    \end{align*}
    and 
    \begin{align*}
        \Delta_w &= \E\bigg[\psi(\Tilde{Y}_i, \Tilde{X}_i, \Tilde{\gamma}_w^{*}, \Tilde{\beta}_w^{*}) \cdot \psi(\Tilde{Y}_i, \Tilde{X}_i, \Tilde{\gamma}_w^{*}, \Tilde{\beta}_w^{*})' \mid \Tilde{W}_i=w \bigg] \\
        &= \E\bigg[(\Tilde{Y}_i-\Tilde{\gamma}_w^{*}-(\Tilde{X}_i - \mu_{\Tilde{X}}) \Tilde{\beta}_w^{*})^2 \cdot \begin{pmatrix}
            -1 \\
             (\Tilde{X}_i - \mu_{\Tilde{X}})'
        \end{pmatrix}
        \begin{pmatrix}
            -1 \\
             (\Tilde{X}_i - \mu_{\Tilde{X}})^{'}
        \end{pmatrix} ^{'}\mid \Tilde{W}_i=w\bigg].
    \end{align*}
    The variance of $\Tilde{\gamma}_w$ is the $(1,1)$ element of the covariance matrix. Because $\Gamma_w$ is block diagonal, the $(1,1)$ element is equal to
    \begin{align*}
        \mathrm{MSE}_w &= \E[(\Tilde{Y}_i-\Tilde{\gamma}_w^{*}-(\Tilde{X}_i - \mu_{\Tilde{X}}) \Tilde{\beta}_w^{*})^2 \mid \Tilde{W}_i=w] \\
        &= \E[(\Tilde{Y}_i - \Tilde{\alpha}_{w}^{*} - \Tilde{X}_i^{'}\Tilde{\beta}^{*}_{w})^2 \mid \Tilde{W}_i=w].
    \end{align*}
    Therefore, we have 
    \begin{align*}
        \sqrt{N_w}\begin{pmatrix}
            \Tilde{\gamma}_w - \Tilde{\gamma}_w^{*}\\
            \Tilde{\beta}_w - \Tilde{\beta}_w^{*}
        \end{pmatrix} \convd \mathrm{N} 
        \begin{pmatrix}
            \begin{pmatrix}
                0\\
                0
            \end{pmatrix},
            \mathrm{MSE}_w
            \begin{pmatrix}
            1 & 0\\
            0 & \big(\E[(\Tilde{X}_i - \mu_{\Tilde{X}})^{'} (\Tilde{X}_i - \mu_{\Tilde{X}})]\big)^{-1}
        \end{pmatrix}
        \end{pmatrix},
    \end{align*}
    which implies
    \begin{equation}
    \label{eq:convd_gamma}
        \sqrt{N} (\Tilde{\gamma}_{(w)} - \Tilde{\gamma}^{*}_{(w)}) \convd \mathrm{N}\bigg(0, \frac{\mathrm{MSE}_w}{\pr(\Tilde{W}_i=w)}\bigg).
    \end{equation}
    
    As shown before, $\tau= C_{p,\epsilon_w}(\Tilde{\gamma}_{1}^{*}-\Tilde{\gamma}_{0}^{*})$. Also,  $ C_{p,\epsilon_w}\Tilde{\tau}_{OLS}= C_{p,\epsilon_w}(\Tilde{\gamma}_1-\Tilde{\gamma}_0)= C_{p,\epsilon_w}\{\Tilde{\alpha}_1-\Tilde{\alpha}_0 + \Bar{\Tilde{X}}(\Tilde{\beta}_1-\Tilde{\beta}_0 )\}$ is the consistent estimator for $\tau$. Noting that $\Tilde{\beta}_1$, $\Tilde{\beta}_0$, $\Tilde{\gamma}_1$ and $\Tilde{\gamma}_0$ are all asymptotically independent, the asymptotic distribution of $\Tilde{\tau}_{OLS}$ is expressed as
    \begin{equation*}
        \sqrt{N}(C_{p,\epsilon_w}\hat{\tau}_{OLS} - \tau) \convd \mathrm{N} \bigg(0, C_{p,\epsilon_w}^2\bigg(\frac{\mathrm{MSE}_1}{\rho_1}+\frac{\mathrm{MSE}_0}{\rho_0}\bigg) \bigg).
    \end{equation*}
\end{proof}

\subsection{Simulation setups for Regression Adjustment}

We empirically evaluate the frequentist properties of the OLS estimator developed in Section \ref{sec:regression_adjustment}. We consider the joint privacy mechanism in Section \ref{sec:jointprivacymech} and use the same data-generating mechanisms in Section \ref{sec:simulationstudies}. We release $X_{i,d}$ after applying the Laplace mechanism. Specifically, the generated covariates satisfy the following sensitivity: $\Delta_{X}=3$.
Accordingly, we add the Laplace noise $\mathrm{Lap}(3/\epsilon_y)$ to $X_{i,k}$ for $k=1,2,3$. Then, we obtain the private data $\Tilde{X}_{i,k}, \Tilde{Y}_i, \Tilde{W}_i$. By composition, this privacy mechanism guarantees that $(\Tilde{Y}_i, \Tilde{W}_i)$ satisfies $(\epsilon_y + \epsilon_w)$-DP and $(\Tilde{X}_{i,k}, \Tilde{Y}_i, \Tilde{W}_i)$ satisfies $(\epsilon_x + \epsilon_y + \epsilon_w)$-DP.

\subsection{Results}

\begin{table*}
	\centering
	\caption{Evaluation metrics for the na\"ive and OLS estimators ($N=1000, N_{sim}=2000$) under the joint scenario. $N_{sim}$ denotes the number of simulations. $\epsilon_{\mathrm{tot}}$ denotes the total privacy budget, $\epsilon_{\mathrm{tot}}=\epsilon_{x}+\epsilon_{y}+\epsilon_{w}$.}
	\begin{adjustbox}{width=14.cm}
		\begin{tabular}{rrrrrrrrrr}
			\toprule
			 & &\multicolumn{2}{c}{Coverage}   & \multicolumn{2}{c}{Bias}   & \multicolumn{2}{c}{MSE}   & \multicolumn{2}{c}{Interval Width} \\
			\cmidrule(lr){3-4} \cmidrule(lr){5-6} \cmidrule(lr){7-8} \cmidrule(lr){9-10} 
			 $\epsilon_{\mathrm{tot}}$ & $(\epsilon_x,\epsilon_y,\epsilon_w)$  & Na\"ive & OLS & Na\"ive & OLS & Na\"ive & OLS & Na\"ive & OLS \\ \hline
			3  & $(1,1,1)$         & $95.3\%$ & $95.7\%$ & $-0.00266$ & $-0.00329$ & $0.0405$ & $0.0371$ & $0.798$ & $0.770$ \\ 
			9  & $(3,3,3)$         & $95.4\%$ & $96.4\%$ & $-0.00105$ & $-0.000422$ & $0.00208$ & $0.00126$ & $0.181$ & $0.142$ \\ 
			30 & $(10,10,10)$      & $95.0\%$ & $96.8\%$ & $-0.000547$ & $-0.000282$ & $0.000906$ & $0.000177$ & $0.120$ & $0.058$ \\ 
			0.3& $(0.1,0.1,0.1)$   & $95.5\%$ & $95.5\%$ & $-0.129$ & $-0.128$ & $0.989$ & $0.984$ & $1.909$ & $1.910$ \\ 
			3  & $(2,0.5,0.5)$     & $94.5\%$ & $94.5\%$ & $-0.00703$ & $-0.00837$ & $0.378$ & $0.375$ & $1.748$ & $1.749$ \\ 
			3  & $(0.5,2,0.5)$     & $95.2\%$ & $95.4\%$ & $-0.00576$ & $-0.00406$ & $0.0484$ & $0.0373$ & $0.857$ & $0.754$ \\ 
			3  & $(0.5,0.5,2)$     & $95.4\%$ & $95.0\%$ & $-0.00238$ & $-0.00263$ & $0.0575$ & $0.0565$ & $0.929$ & $0.923$ \\ 
			3  & $(0.5,1.25,1.25)$ & $94.6\%$ & $94.5\%$ & $0.00480$ & $0.00276$ & $0.0210$ & $0.0187$ & $0.547$ & $0.518$ \\ 
			3  & $(1.25,0.5,1.25)$ & $95.3\%$ & $95.2\%$ & $-0.00101$ & $-0.00246$ & $0.103$ & $0.101$ & $1.232$ & $1.225$ \\ 
			3  & $(1.25,1.25,0.5)$ & $94.6\%$ & $95.7\%$ & $0.00137$ & $0.00150$ & $0.102$ & $0.0889$ & $1.195$ & $1.144$ \\ 
			\bottomrule
		\end{tabular}
	\end{adjustbox}
	\label{tab:simulation_freq_N1000}
\end{table*}

Tables \ref{tab:simulation_freq_N1000} and \ref{tab:simulation_freq_N10000} present the performance evaluation of the na\"ive and OLS estimators for $N=1000,10000$ with various privacy budgets for $\epsilon_x$, $\epsilon_y$ and $\epsilon_w$. We let $\epsilon_{tot} = \epsilon_x + \epsilon_y + \epsilon_w$.  Both estimators achieve about $95\%$ coverage for $N=1000,10000$ as expected. For bias and MSE, we observe smaller bias and MSE for larger privacy budgets. For the same levels of privacy budgets, both bias and MSE improve when $N$ increases, which empirically supports our consistency and asymptotically unbiased properties of the estimators. 

When we have a tight privacy budget of $(\epsilon_x,\epsilon_y,\epsilon_w)=(0.1,0.1,0.1)$, the length of the confidence interval of the frequentist estimators is nearly $2$, which is almost non-informative about the estimand. 
When $N$ increases, the interval length gets smaller and becomes informative enough for some allocations, e.g., $(\epsilon_x,\epsilon_y,\epsilon_w)=(1.25,0.5,1.25)$. However, with strict budget constraints and a small sample size, the analysis results may tell us little about the estimands, even though their consistency and confidence intervals are statistically valid. This is an inevitable trade-off between privacy protection and the accuracy of the results. 

\subsection{Discussions}
\label{sec:discussions}

\begin{table*}
	\centering
	\caption{Evaluation metrics for the na\"ive and OLS estimators ($N=10000, N_{sim}=2000$) under the joint scenario.}
	\begin{adjustbox}{width=14.cm}
		\begin{tabular}{rrrrrrrrrr}
			\toprule
			 & &\multicolumn{2}{c}{Coverage}   & \multicolumn{2}{c}{Bias}   & \multicolumn{2}{c}{MSE}   & \multicolumn{2}{c}{Interval Width} \\
			\cmidrule(lr){3-4} \cmidrule(lr){5-6} \cmidrule(lr){7-8} \cmidrule(lr){9-10} 
			 $\epsilon_{\mathrm{tot}}$ & $(\epsilon_x,\epsilon_y,\epsilon_w)$  & Na\"ive & OLS & Na\"ive & OLS & Na\"ive & OLS & Na\"ive & OLS \\ \hline
			3  & $(1,1,1)$         & $95.4\%$ & $95.2\%$ & $-0.00174$ & $-0.00196$ & $0.00407$ & $0.00376$ & $0.252$ & $0.243$ \\ 
			9  & $(3,3,3)$         & $94.7\%$ & $94.7\%$ & $-0.000154$ & $-0.000149$ & $0.000216$ & $0.000136$ & $0.0573$ & $0.0454$ \\ 
			30 & $(10,10,10)$      & $94.6\%$ & $96.3\%$ & $0.000213$ & $-0.0000316$ & $0.0000962$ & $0.0000184$ & $0.0380$ & $0.0183$ \\ 
			0.3& $(0.1,0.1,0.1)$   & $94.4\%$ & $94.3\%$ & $-0.104$ & $-0.101$ & $0.919$ & $0.921$ & $1.883$ & $1.885$ \\ 
			3  & $(2,0.5,0.5)$     & $94.9\%$ & $95.1\%$ & $-0.00380$ & $-0.00358$ & $0.0535$ & $0.0520$ & $0.915$ & $0.906$ \\ 
			3  & $(0.5,2,0.5)$     & $95.7\%$ & $95.7\%$ & $0.00112$ & $0.000358$ & $0.00466$ & $0.00356$ & $0.271$ & $0.237$ \\ 
			3  & $(0.5,0.5,2)$     & $95.9\%$ & $95.9\%$ & $0.000703$ & $0.000989$ & $0.00524$ & $0.00512$ & $0.295$ & $0.292$ \\ 
			3  & $(0.5,1.25,1.25)$ & $95.9\%$ & $95.9\%$ & $0.00133$ & $0.00124$ & $0.00187$ & $0.00170$ & $0.173$ & $0.163$ \\ 
			3  & $(1.25,0.5,1.25)$ & $95.1\%$ & $95.0\%$ & $-0.000968$ & $-0.000691$ & $0.0106$ & $0.0104$ & $0.405$ & $0.401$ \\ 
			3  & $(1.25,1.25,0.5)$ & $95.4\%$ & $95.4\%$ & $0.00247$ & $0.00279$ & $0.00957$ & $0.00848$ & $0.391$ & $0.369$ \\ 
			\bottomrule
		\end{tabular}
	\end{adjustbox}
	\label{tab:simulation_freq_N10000}
\end{table*}

In the simulations, we consider different divisions of the same overall privacy budget, $\epsilon_{tot}=3$, which suggests an allocation strategy of the budget. Among all the budget allocations with $\epsilon_{tot}=3$, we see that $(\epsilon_x,\epsilon_y,\epsilon_w)=(0.5,1.25,1.25)$ achieves the lowest MSE for both na\"ive and OLS estimators. Thus, it seems reasonable to assign a strict budget to $X$, and larger budgets to $Y$ and $W$. We also see that for most allocations with budgets $\epsilon_{tot}\leq 3$, there is minimal gain in MSE for the OLS over the na\"ive estimator. However, for  $(\epsilon_x,\epsilon_y,\epsilon_w)=(10,10,10),(3,3,3),(0.5,2,0.5)$,  we see that the OLS estimator does significantly outperform the na\"ive estimator in terms of MSE. This result follows from the fact that the regression adjustment technique in randomized experiments \citep{Freedman2008,Lin2013} helps reduce the variance of the OLS estimator, leading to better MSE. Intuitively, the regression adjustment works for $(\epsilon_x,\epsilon_y,\epsilon_w)=(10,10,10)$ because the privatized data contains smaller noise, and $\Tilde{X}$ still contains some information to explain $\Tilde{Y}$. When the total budget is smaller ($\epsilon_{\mathrm{tot}}\leq 3$), however, the gain is limited. 

We here further discuss some limitations to the gains in precision of the estimator for the PATE from including covariates from theoretical perspectives.
In large samples, including covariates in the regression function under usual randomized experiments will not lower the precision \citep{imbens_rubin_2015}. However, DP mechanisms under randomization pose unique challenges. First, $\mathrm{MSE}_w$ in Theorem \ref{thm:properties_of_ols} can be written as follows:
\begin{equation}
\label{eq:mse_decomp}
    \mathrm{MSE}_w 
    = \Var[Y_i|\Tilde{W}_i=w] + \E[Y_i|\Tilde{W}_i=w]^2 + \frac{1}{\epsilon_y^2} - \E[\Tilde{Y}_i' \Tilde{X}_i (\Tilde{X}_i'\Tilde{X}_i)^{-1}\Tilde{X}_i' \Tilde{Y}_i |\Tilde{W}_i=w].
\end{equation}
The last term, $\E[\Tilde{Y}_i' \Tilde{X}_i (\Tilde{X}_i'\Tilde{X}_i)^{-1}\Tilde{X}_i' \Tilde{Y}_i|\Tilde{W}_i=w]$, is effectively the gain in precision from including covariates. This term implies that the gain is zero when $\Tilde{X}_i$ and $\Tilde{Y}_i$ are orthogonal, but is always positive otherwise. As adding large independent noise to $X_i$ and $Y_i$ makes the privatized observations less correlated, the gain becomes negligible when $\epsilon_x$ and $\epsilon_y$ are small. We also note that the first two terms in \eqref{eq:mse_decomp} are bounded due to the sensitivity of $Y$; however, the last two terms are unbounded, making them the dominant precision factors, especially when $\epsilon_x$ and $\epsilon_y$ are small. Therefore, the gain from adding covariates in inference is actually limited in our LDP scenarios. 

\end{document}